\documentclass[10pt,conference]{IEEEtran}
\input{psfig.sty}
\input{epsf.sty}
\usepackage{hyperref} 
\usepackage{fancybox}   
\usepackage{psfig}      
\usepackage{graphicx}
\usepackage{amsmath}
\usepackage{color}
\usepackage{epsfig}
\usepackage{amssymb}
\usepackage{verbatim}
\usepackage{enumitem}
\usepackage{bbm}

\IEEEoverridecommandlockouts

\newcommand{\remove}[1]{}
\newcommand{\argmin}{\operatornamewithlimits{argmin}}

\newtheorem{thm}{Theorem}

\newtheorem{lem}{Lemma}

\newcommand{\qed}{\hfill \ensuremath{\Box}}

\usepackage{setspace}
\usepackage{subfigure}

\begin{document}

\title{As-You-Go Deployment of a Wireless Network with On-Line Measurements and Backtracking \thanks{The research 
reported in this paper was supported by a Department of Electronics and Information 
Technology (DeitY, India) and NSF (USA) funded project on Wireless 
Sensor Networks for Protecting Wildlife and Humans, by an Indo-French Centre for 
Promotion of Advance Research (IFCPAR) funded project, and by the Department of Science and Technology (DST, India). }}


\newcounter{one}
\setcounter{one}{1}
\newcounter{two}
\setcounter{two}{2}



\author{
Arpan~Chattopadhyay$^\fnsymbol{one}$, Marceau~Coupechoux$^\fnsymbol{two}$, and Anurag~Kumar$^\fnsymbol{one}$\\
\parbox{0.49\textwidth}{\centering $^\fnsymbol{one}$Dept. of ECE, Indian Institute of Science\\
Bangalore 560012, India\\
arpanc.ju@gmail.com, anurag@ece.iisc.ernet.in}
\hfill
\parbox{0.49\textwidth}{\centering $^\fnsymbol{two}$Telecom ParisTech and CNRS LTCI \\
Dept. Informatique et R\'eseaux\\
23, avenue d'Italie, 75013 Paris, France\\
marceau.coupechoux@telecom-paristech.fr}
}

\maketitle
\thispagestyle{empty}

\begin{abstract}
We are motivated by the need, in some applications, for impromptu or as-you-go deployment of wireless sensor networks. 
A person walks  along a line, 
making link quality measurements with the previous relay at equally spaced locations, and deploys relays at some 
of these locations, so as to connect a sensor placed on the line with a
sink at the start of the line. 
In this paper, we extend our earlier work on the problem  
(see \cite{chattopadhyay-etal13measurement-based-impromptu-placement_wiopt}) to incorporate two new aspects: 
(i) inclusion of path outage in the deployment objective, and (ii) permitting the deployment agent to make 
measurements over several consecutive steps before selecting a placement location among them 
(which we call \emph{backtracking}). We consider a light traffic regime, and
formulate the problem as 
a Markov decision process. Placement algorithms are obtained for two cases: (i) the distance to the source is geometrically 
distributed with known mean, 
and (ii) the average cost per step case. 
We motivate the per-step cost function in terms of several known forwarding protocols 
for sleep-wake cycling wireless sensor networks. We obtain the structures of the 
optimal policies for the various formulations, and provide some sensitivity results about 
the policies and the optimal values. We then provide a numerical study of the algorithms, 
thus providing insights into the advantage of backtracking, and a comparison with simple heuristic placement policies.
\end{abstract}

\section{Introduction}
\label{Introduction}

Wireless interconnection of resource-constrained mobile user devices or wireless sensors to
the wireline infrastructure via relay nodes is an important
requirement, since a direct one-hop link from the source node to 
the infrastructure ``base-station" may not always be feasible, due to distance or poor channel condition. 
Such wireless interconnection of sensors with the wireline infrastructure is usually 
performed by a multi-hop wireless network, the resulting system being commonly called a Wireless Sensor Network (WSN). 
There are situations in which a WSN needs to be deployed in an ``as-you-go" fashion. One such situation is in 
emergencies, e.g, situational awareness networks deployed by first-responders such as fire-fighters 
or anti-terrorist squads. As-you-go deployment is also of interest when deploying networks 
over large terrains, such as forest trails, particularly when the network is temporary and needs to be quickly 
redeployed in a different part of the forest (e.g., to monitor a moving phenomenon such as groups of wildlife), 
or when the deployment needs to be stealthy (e.g., to monitor fugitives). 
Motivated by these more general problems, we consider the problem of ``as-you-go" deployment of relay nodes
along a line, between a sink node and a source node (see
Figure~\ref{fig:line-network-general}), where the {\em deployment operative}\footnote{In this paper 
we consider a \emph{single} person carrying out the deployment and refer to this person 
as a ``deployment operative," or a ``deployment agent."}
starts from the sink node, places relay nodes along the line, and
places the source node where the line ends.  

\begin{figure}[!t]
\centering
\includegraphics[scale=0.37]{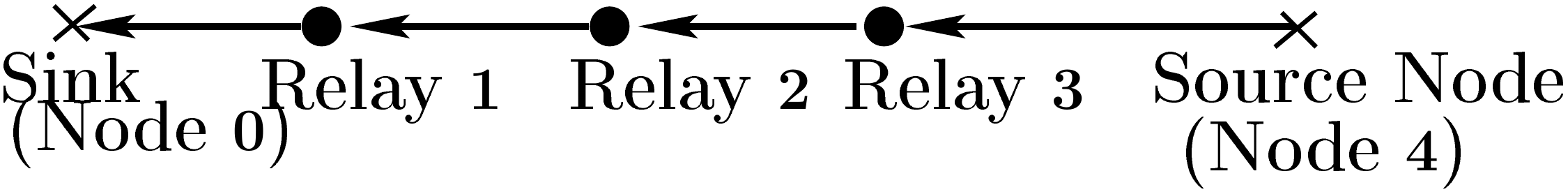}
\vspace{-2mm}
\caption{A line network with one source, one sink and three relays.}
\label{fig:line-network-general}
\vspace{-5mm}
\end{figure}

In \cite{chattopadhyay-etal13measurement-based-impromptu-placement_wiopt} we have 
formulated such a problem as one of optimal sequential relay 
deployment driven by measurements between a node yet to be deployed and the last relay already deployed. 
In \cite{chattopadhyay-etal13measurement-based-impromptu-placement_wiopt} we worked 
under the restriction that the deployment agent only moves forward. 
Such forward-only movement would be a necessity if the deployment needs to be quick. Due to shadowing, 
the path-loss over a link of a given length is random and a more efficient deployment can be expected if 
link quality measurements at several  locations along the line are compared and an optimal choice made among 
these. Since, in general, this would require the deployment agent to retrace his steps, we call this approach 
\emph{backtracking}. Backtracking would take more time, but might provide a good compromise between deployment 
speed and deployment efficiency, for an application such as the as-you-go deployment of a wireless sensor 
network over a forest trail (e.g, for wildlife monitoring).
When placing a relay at some distance from the previous relay, we can expect a better deployment if we explore several locations 
in the vicinity, at which the link qualities can be expected to be uncorrelated, and then pick the best among these.

In this paper, we mathematically formulate the  problems of
as-you-go deployment of relays along a line as optimal sequential decision problems. 
We introduce various measures of hop-cost with justification, and then formulate relay placement problems that 
minimize (i) the expected total hop cost when the distance $L$ of the source from the sink is geometrically distributed, 
or (ii) the expected average cost per-step. 
Our channel model accounts for path-loss, 
shadowing, and fading. The cost 
of a deployment is evaluated as a linear combination of three components: the sum or the maximum transmit 
power along the path, the sum outage 
probability along the path, and the number of relays deployed. 
We explore deployment with and without backtracking. 
We formulate each of these problems as a Markov decision process (MDP), obtain the optimal policy structures, illustrate their performance 
numerically 
and compare their performance with reasonable heuristics.

\subsection{Related Work}\label{subsec:related_work}

Recent years have seen increasing interest 
in the research community to explore the impromptu relay placement problem.  
For example, Howard et al., in \cite{howard-etal02incremental-self-deployment-algorithm-mobile-sensor-network}, 
provide heuristic algorithms for incremental deployment of
sensors in order to cover the deployment area. 
Souryal et al., in \cite{souryal-etal07real-time-deployment-range-extension}, address the
problem of impromptu wireless network deployment with experimental study of indoor RF link quality variation; 
similar approach is taken in \cite{aurisch-tlle09relay-placement-emergency-response} also. 
The authors of \cite{liu-etal10breadcrumb} describe a {\em breadcrumbs} system for aiding firefighting inside buildings. 
However,these approaches are based purely on heuristics and experimentation; 
they lack the rigour, both in the 
formulation and in the deployment strategy, and hence are not convincingly optimal or near optimal in terms of performance. 
In our work our effort has been to formulate these problems as optimal sequential 
decision problems in order derive optimal policies whose performance can be compared 
with simple heuristics, and which could be used to propose other heuristics. 
Recently, Sinha et al. (\cite{sinha-etal12optimal-sequential-relay-placement-random-lattice-path}) 
have provided an algorithm based on MDP formulation in order to establish a multi-hop 
network between a sink and an unknown source location, 
by placing relay nodes along a random lattice path. Their model uses a deterministic mapping between power and wireless 
link length, and, hence, does not consider the effect of shadowing that leads to statistical variability of the transmit power 
required to maintain the link quality over links having the same length. This problem was addressed by 
Chattopadhyay et al. in \cite{chattopadhyay-etal13measurement-based-impromptu-placement_wiopt}, where they considered 
spatial variation in link qualities due to shadowing in the context of as-you-go deployment along a 
line of unknown random length. The variation of link qualities over space led to the introduction of measurement-based 
deployment, in which the deployment agent measures the power required to establish a link (with a given quality) 
to the already placed nodes; the placement algorithm uses this value to decide whether to place a relay at that point.

The work reported in \cite{chattopadhyay-etal13measurement-based-impromptu-placement_wiopt}, 
however, has limitations that we address in the present paper. 

\begin{enumerate}[label=(\roman{*})]
 \item The framework of \cite{chattopadhyay-etal13measurement-based-impromptu-placement_wiopt} 
requires the link of each hop to have a certain outage probability. 
In practice, as the deployment agent walks away from the previously placed node, he can reach a 
point where even the maximum node power does not provide a link of the desired quality to the previous relay, 
and walking any farther is unlikely to provide a workable link. At this point the deployment is 
considered to have failed. In our present paper, we do not bound the outage probability of each hop but make 
the sum outage over all the hops a part of the optimization objective. 

\item In the framework of \cite{chattopadhyay-etal13measurement-based-impromptu-placement_wiopt}, 
the deployment agent can only move forward. 
In the present paper we introduce ``backtracking," which permits the deployment agent 
to compare the link qualities over several potential placement locations before deploying the relay at any one of them.
\end{enumerate}

\subsection{Organization}\label{subsec:organization}

The rest of the paper is organized as follows. The system model and notation has been described in 
Section~\ref{sec:system_model_and_notation}.  
As-you-go deployment (without backtracking, for a line having geometrically distributed length) for sum and max power 
objectives have been described in Section~\ref{sec:sum_power_sum_outage_no_backtracking_discounted} and 
Section~\ref{sec:max_power_sum_outage_no_backtracking_discounted} respectively. 
Section~\ref{sec:sum_power_sum_outage_with_backtracking_discounted} and 
Section~\ref{sec:max_power_sum_outage_with_backtracking_discounted} have addressed the problems of 
as-you-go deployment with backtracking, for a line having geometrically distributed length, for sum and max power 
objectives respectively. 
As-you-go deployment with backtracking along a line of infinite length, for average cost per step 
objective, has been discussed in Section~\ref{sec:backtracking_average_cost}. 
The numerical results have been 
discussed in Section~\ref{sec:numerical_work}, followed by the conclusion in Section~\ref{sec:conclusion}.

\section{System Model and Notation}\label{sec:system_model_and_notation}

\subsection{Length of the Line}\label{subsection:length_of_the_line}

We consider two different models:
\begin{enumerate}[label=(\roman{*})]
 \item  We first consider the scenario where the distance $L$ to the source from the sink 
at the start of the line is a priori unknown, but there is prior 
information (e.g., the mean $\overline{L}$) that leads us to model $L$ as a geometrically 
distributed number of steps. The step length $\delta$ 
and the mean $\overline{L}$, can be used 
to obtain the parameter of the geometric distribution, i.e., the probability $\theta$ that the line ends at the 
next step. 
All distances are assumed to be integer multiples of $\delta$. 
In the problem formulation, we assume $\delta=1$ for simplicity.\footnote{The geometric distribution 
is the maximum entropy discrete probability 
mass function with a given mean. Thus, by using the geometric distribution, 
we are leaving $L$ as uncertain as we can, given the prior knowledge of $\overline{L}$.}

\item  Next, we consider the setting where the line has infinite length. This can be useful 
where the line is long enough, and the end is not known (e.g., a long forest trail). Also, it can be used to deploy 
 a chain of relays over a long line, which can be used by several source-sink pairs, and the source-sink pairs 
could even be moved from one place to another (if required).
\end{enumerate}

\subsection{Channel Model}\label{subsection:channel_model}

In order to model the wireless channel, we consider the usual 
aspects of path-loss, shadowing, and fading. The received power for a particular link (i.e., a transmitter-receiver pair)  
of length $r$ is given by:
\begin{equation}
 P_{rcv}=P_T c \bigg(\frac{r}{r_0×}\bigg)^{-\eta}HW \label{eqn:channel_model}
\end{equation}
where $P_T$ is the transmit power, $c$ corresponds to the path-loss at the reference distance $r_0$, 
$\eta$ is the path-loss exponent, $H$ denotes the fading random variable (e.g., it could be an exponential 
random variable) and $W$ 
denotes the shadowing random variable. 
$H$ accounts for the variation of the received power over time, and it takes  
independent values from one coherence time to another. 
The path-loss between a transmitter and a receiver at a given distance can have a large spatial  
variability around the mean path-loss (averaged over fading), as the transmitter is moved over different points at the 
same distance; this is called shadowing.\footnote{Consider (\ref{eqn:channel_model}). If we transmit a 
sufficiently large number of packets 
on a link over multiple coherence times and record the received signal strength of all the packets, we can compute 
$\overline{P}_{rcv}$ which is the mean received signal power averaged over fading. 
But if the realization of shadowing in that link is $w$, 
then $\overline{P}_{rcv}=P_T c \bigg(\frac{r}{r_0×}\bigg)^{-\eta}w \mathbb{E}(H)$, 
from which we can easily calculate $w$.}  Shadowing is usually modeled as a log-normally distributed, 
random, multiplicative path-loss factor; in dB, shadowing is normally distributed with values of standard 
deviation as large as $8$ to $10$~dB. Also, shadowing is spatially uncorrelated over distances that depend on 
the sizes of the objects in the propagation environment (see \cite{agarwal-patwari07correlated-shadow-fading-multihop}); 
our measurements in a 
forest-like region of our Indian Institute of Science campus gave a decorrelation distance of $6$~meters. This is evident from 
Figure~\ref{fig:correlation_vs_distance}, where the variation of the measured correlation, $\rho$, between shadowing 
of two links (having  
one end common and the other ends located on the same straight line) with the distance between the other two ends, $d$, 
has been shown. Log-normality of shadowing was verified via the Kolmogorov-Smirnov goodness-of-fit test, and hence, 
the very small values of $\rho$ at $d \geq 6$~meters show that we can safely assume independent shadowing 
at different potential relay locations  
if the step size $\delta$ is $6$~meters or above. However, in our formulation, 
$W$ is assumed to take values from a finite set $\mathcal{W}$ with the probability mass function $p_W(w):=\mathbb{P}(W=w)$; 
in our numerical work we have quantized the range of values that log-normal shadowing can assume. 

A link is considered to be in \emph{outage} if  the received 
signal power drops (due to fading) below $P_{rcv-min}$  (e.g., below $-88$~dBm). 
Since practical radios can only be set to transmit at a finite set of power levels, 
the transmit power of each node can be chosen from a discrete set, $\mathcal{S}:=\{P_1, P_2, \cdots, P_M \}$, where 
$\{P_1, P_2, \cdots, P_M \}$ is arranged in ascending order. 
For a link of length $r$, a transmit power $\gamma\in \mathcal{S}$ and any particular realization of shadowing $W=w$, 
the outage probability is denoted by $P_{out}(r,\gamma,w)$, which is increasing in $r$ and decreasing 
in $\gamma$, $w$ (according to (\ref{eqn:channel_model})). 
It is also easy to check that the random variable $P_{out}(r,\gamma,W)$ is 
stochastically increasing in $r$ for each $\gamma \in \mathcal{S}$, and stochastically 
decreasing in $\gamma$ for each $r$. Note that $P_{out}(r,\gamma,w)$ depends on the 
fading statistics in the environment.\footnote{For a link with shadowing realization $w$, if the transmit power 
is $\gamma$, the received power of a packet will be $P_{rcv}=\gamma c \bigg(\frac{r}{r_0×}\bigg)^{-\eta}wH$. 
Outage is defined to be 
the event $P_{rcv} \leq P_{rcv-min}$. If $H$ is exponentially distributed with mean $1$, then we have, 
$P_{out}(r,\gamma,w)=\mathbb{P}\bigg( \gamma c \bigg(\frac{r}{r_0×}\bigg)^{-\eta}wH \leq P_{rcv-min} \bigg)=
1-e^{-\frac{P_{rcv-min}(\frac{r}{r_0×})^{\eta}}{\gamma c w×}}$.} 

\begin{figure}[!t]
\centering
\includegraphics[scale=0.32]{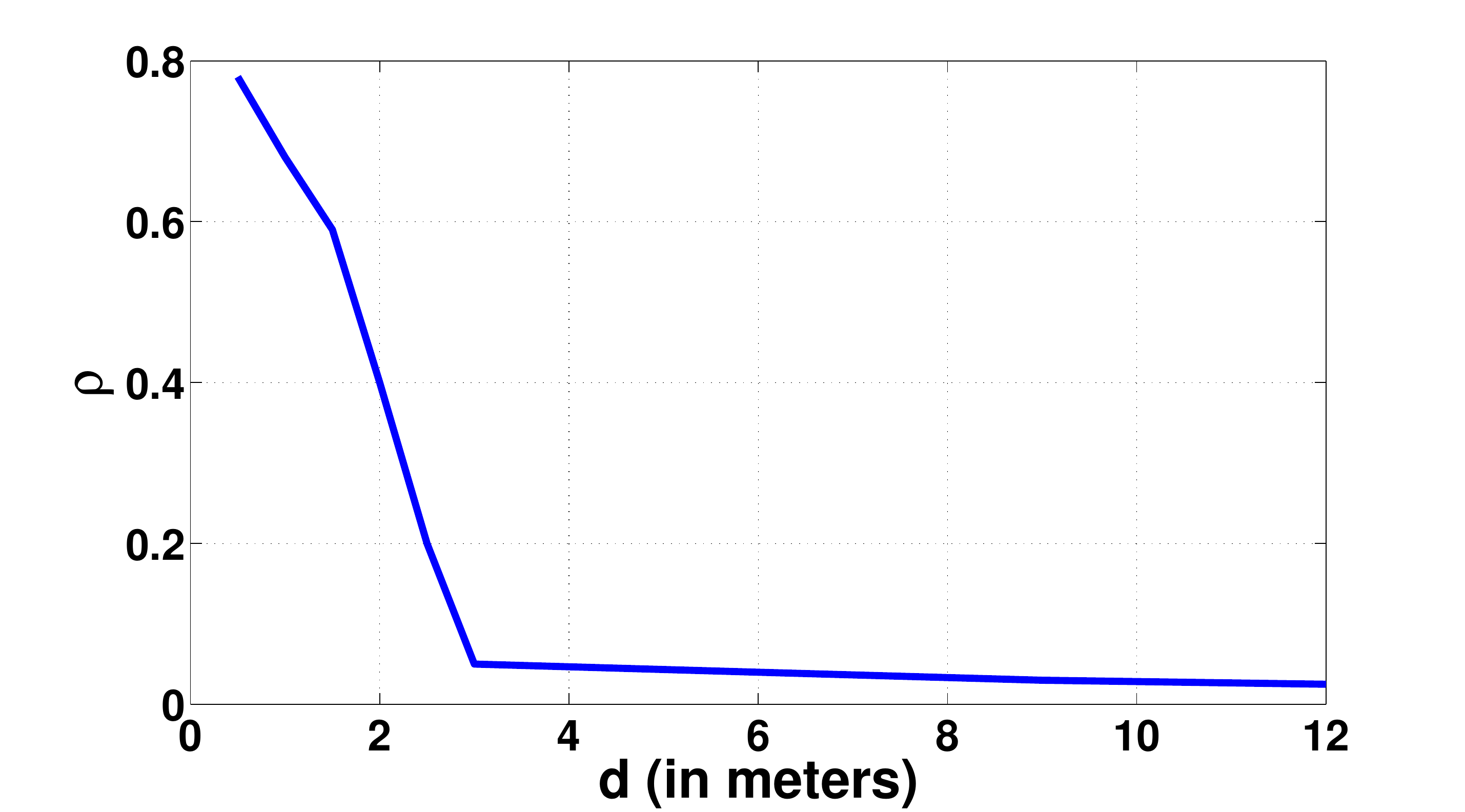}
\vspace{-2mm}
\caption{Variation of link shadowing correlation $\rho$ as a function of distance $d$ measured in a forest-like 
environment in our campus; one end (either the 
transmitter or the receiver) is common to both links and the other ends for both links are kept on the same line, but $d$ distance 
apart from each other.}
\label{fig:correlation_vs_distance}
\vspace{-3mm}
\end{figure}

\subsection{Deployment Process and Some Notation}\label{subsection:deployment_process_notation}

\begin{figure}[!t]
\centering
\includegraphics[scale=0.32]{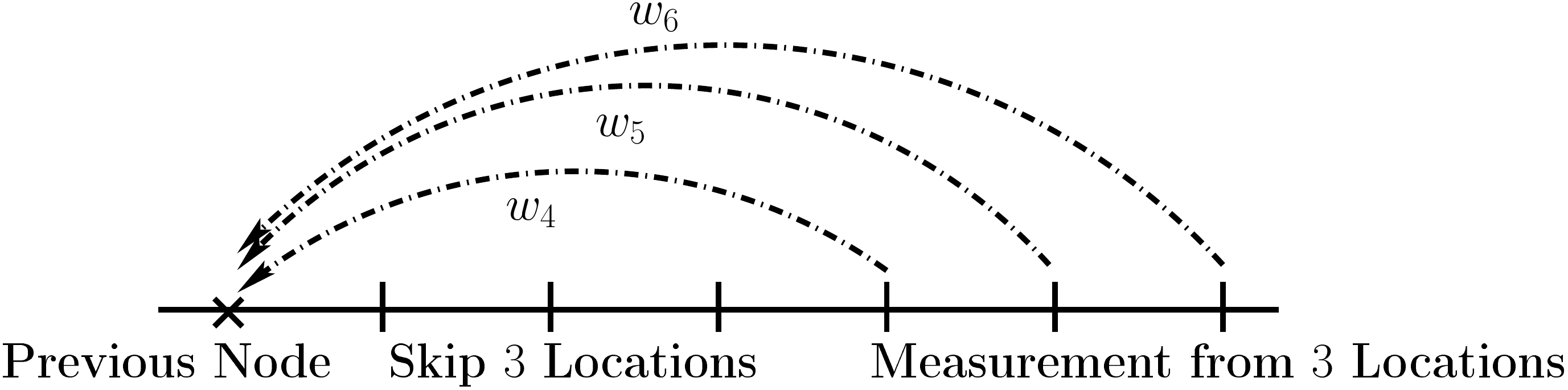}
\vspace{-2mm}
\caption{Backtracking with $A=3$ and $B=3$; the deployment agent skips the first $A$ steps from the previous node 
and measures the shadowing $w_{A+1},w_{A+2},\cdots,w_{A+B}$ from next $B$ locations in order to 
decide where to place the next relay.}
\label{fig:backtracking_illustration}
\vspace{-5mm}
\end{figure}

As the deployment agent walks along the line, at each step or at some subset of steps (each step is assumed to be 
a potential relay location, thereby discretizing the problem) he measures the link quality from the 
current location to the previous node (see Figure~\ref{fig:backtracking_illustration}); 
these measurements are used to decide where to place the next relay node and at what transmit power it should operate. 
In this paper, we do not consider the possibility of
another person following behind, who can learn from the measurements
and actions of the first person, thereby supplementing
the actions of the preceding individual. 
For deployment with {\em backtracking}, we assume that after placing a node, 
the deployment agent skips the next $A$ locations (i.e., walks forward a distance $A \delta$, where $A \geq 0$) 
and measures the shadowing $\underline{w}:=(w_{A+1},w_{A+2},\cdots,w_{A+B})$\footnote{Underlined symbols denote vectors 
in this paper.} to the previous node from locations 
$(A+1),(A+2), \cdots,(A+B)$. 
Then he places the relay at one of the locations $(A+1),(A+2), \cdots,(A+B)$ and moves on. This procedure 
is illustrated in Figure~\ref{fig:backtracking_illustration}. 
For the geometrically distributed length if the line ends within $(A+B)$ steps from the previous node, 
then the source is placed where the line ends. In this case, after the deployment process is complete 
(i.e., when the source is placed), we denote the 
number of deployed relays by $N$, which is a random number, with the randomness coming 
from the randomness in the link qualities (due to shadowing) and in the length of the line. 

As shown in Figure~\ref{fig:line-network-general}, the sink is called Node $0$, the relay 
closest to the sink is called Node $1$, and finally the source is called Node $(N+1)$. 
The link whose transmitter is Node $i$ and receiver is Node 
$j$ is called link $(i,j)$. A generic link is denoted by $e$.

We assume that the shadowing at any two different links in the network 
are independent, i.e., $W_{(e_1)}$ is independent of $W_{(e_2)}$ for $e_1 \neq e_2$. This can be a reasonable assumption 
if $\delta$ is chosen to be at least the de-correlation distance of shadowing (as discussed in 
Section~\ref{subsection:channel_model}).

For comparison, we also consider the case in which backtracking is \emph{not} allowed. 
In this case, after placing a relay, the agent skips the next $A$~steps, and 
sequentially measures shadowing from the locations $(A+1),(A+2),\cdots,(A+B)$. As the agent explores the locations 
$(A+1),(A+2),\cdots,(A+B-1)$ and measures the shadowing in those locations, at each step he decides whether to 
place a relay there, and if the decision is to place a relay, then he 
also decides at what transmit power the relay will operate. In this process, if he has walked $(A+B)$ steps away from the previous 
relay, or if he encounters the source location within this distance, then he must place the relay or the source.

The choice of $A$ and $B$ depends on the constraints and requirements for the deployment. A larger value of $A$ will 
result in faster exploration of the line, since many locations can be skipped. For a fixed $A$, a larger value 
of $B$ results in more measurements, and hence we can expect a better performance on an average. However, $A$ and $B$ must be 
chosen such that the outage probability of a randomly chosen link having length $(A+B)$~steps are within tolerable limits with 
high probability.\footnote{Randomness in outage probability of a randomly chosen link 
comes from the spatial variation of link quality 
due to shadowing.}

\subsection{Traffic Model}\label{subsection:traffic_model}

We consider a traffic model where the traffic is so
low that there is only one packet in the network at a time; we call
this the ``lone packet model.''  As a consequence of this assumption,
(i) the performance over each link depends only on the
path-loss, shadowing and fading over that link, as there are no simultaneous
transmissions to cause interference, and (ii) the transmission delay
over each link is easily calculated, as there are no simultaneous
transmitters to contend with. This permits us to easily write down the
communication cost on a path over the deployed relays. 
Such a traffic model is realistic for sensor networks that carry out low duty cycle 
measurement of environment variables, or just carry an occasional alarm packet. 
Also, a design with the lone packet model can be the starting point for a design with desired positive traffic.

\vspace{-3mm}
\subsection{Network Cost}\label{subsection:network_cost}
\vspace{-2mm}
In each case, we evaluate the cost of a deployed network in terms of the sum of certain hop costs. 
In case all the nodes have {\em wake-on radios,} the nodes normally stay in sleep mode, and each sleeping node 
draws a very small current from the battery (see \cite{vodel-hardt13energy-efficient-communication-distributed-embedded-systems}). 
When a node has a packet, it sends a wake-up tone to the intended receiver. 
The receiver wakes up and the sender transmits the 
packet. The receiver sends an ACK packet in reply. Clearly, the energy spent in transmission and reception of data packets 
govern the lifetime of a node, given that the ACK size is negligible compared to the packet size. Also, 
the energy spent in transmission and reception of packets govern the lifetime in certain 
receiver-centric synchronous duty cycled MAC protocols, 
under moderate traffic which ensures no contention and 
substantial amount of energy consumption in data transmission and reception. 

Let $t_p$ be the duration of a packet, and suppose that the node~$i$ uses power $\Gamma_i$ during transmission, 
which can be chosen according to the link quality. 
Let $P_{r}$ denote the power
that any receiving node uses for a packet reception. The 
relay node~$k$ ($1 \leq k \leq N$) can deliver $E/(\Gamma_k+P_{r})t_p$ packets before its battery is drained out. 
The source can deliver 
$E/\Gamma_{N+1}^{'} t_p$ packets, where $\Gamma_{N+1}^{'}$ is the transmit power used by the source. Writing 
$\Gamma_{N+1}^{'}=\Gamma_{N+1}+P_r$, we can write the cost function which
 appropriately captures the lifetime of the network:
\begin{equation}
\mathbb{E} ( \max_{i \in \{1,2,\cdots,N+1\}} \Gamma_i + \xi_o \sum_{i=1}^{N+1}P_{out}^{(i,i-1)}+ \xi_r N )
\label{eqn:cost_function_max_power_sum_outage}
\end{equation}
where $\xi_r$ is the cost of placing a relay and $\xi_o$ is the cost per unit outage probability. 
$P_{out}^{(i,i-1)}$ is the outage probability in the link $(i,i-1)$, and is decreasing in the 
transmit power $\Gamma_i$. The sum outage probability $\sum_{i=1}^{N+1}P_{out}^{(i,i-1)}$ is an indicator of the end-to-end
packet dropping rate when the outage probabilities are small and there is no retransmission for dropped packets.

On the other hand, since the packet arrival rate $\zeta$ at the source is very small, 
the lifetime of the $k$-th node is $T_k:=\frac{E}{\zeta(\Gamma_k+P_{r})t_p×}$ seconds. Hence, the rate at 
which we have to replace the batteries in the network is given by $\sum_k \frac{1}{T_k×}=\sum_k \frac{\zeta(\Gamma_k+P_{r})t_p×}{E}$. The 
energy expenditure due to $P_r$ is absorbed into $\xi_r$, and we have the following cost function:
\begin{equation}
\mathbb{E} ( \sum_{i=1}^{N+1} \Gamma_i + \xi_o \sum_{i=1}^{N+1}P_{out}^{(i,i-1)}+ \xi_r N )
\label{eqn:cost_function_sum_power_sum_outage}
\end{equation}
For average cost 
per step objective, the max power cost does not make any sense and we consider only sum-power cost.

\section{Impromptu Deployment for Geometrically Distributed Length without Backtracking: Sum Power and 
Sum Outage Objective}\label{sec:sum_power_sum_outage_no_backtracking_discounted}

\subsection{Problem Formulation}\label{subsec:problem_formulation_sum_power_sum_outage_no_backtracking}

Here we seek to solve the following problem:

\begin{equation}
 \min_{\pi \in \Pi} \mathbb{E}_{\pi} \bigg( \sum_{i=1}^{N+1}\Gamma^{(i,i-1)}+ 
\xi_o \sum_{i=1}^{N+1}P_{out}^{(i,i-1)} +\xi_r N \bigg)
\label{eqn:sum_power_discounted_no_backtracking}
\end{equation}
where $\Pi$ is the set of all placement policies and $\pi$ is a specific placement policy. 

Let us recall the deployment procedure for no backtracking as described in Section~\ref{subsection:deployment_process_notation}. 
When the agent is $r$ steps away from the previous node ($A+1 \leq r \leq A+B$), 
he measures the shadowing $w$ on the link from the current location to the previous node. 
He uses the knowledge of $(r,w)$ to decide whether to place a node there. 
We formulate (\ref{eqn:sum_power_discounted_no_backtracking}) as a  Markov 
Decision Process with state space $\{A+1,A+2,\cdots,A+B\} \times \mathcal{W}$. 
At state $(r,w), (A+1) \leq r \leq (A+B-1), w \in \mathcal{W}$, the action is either to 
place a relay and select some transmit power $\gamma \in \mathcal{S}$, or not to place. When $r=A+B$, the only 
feasible action is to place and select a transmit power $\gamma \in \mathcal{S}$. Note that, the 
problem restarts after placing a relay, because of the memoryless property of the geometric distribution and 
the independence of shadowing across links; the state of the system at such regeneration points is denoted by $\mathbf{0}$. 
When the source is placed, the process terminates. 
The randomness in the system comes from the geometric distribution of the length of the line and 
the random shadowing in different links. 
{\em Note that the cost function in (\ref{eqn:sum_power_discounted_no_backtracking}) 
can also be motivated as Lagrangian relaxations of constraints on the expectations of the sum outage and 
the number of deployed relays, $N$.}

\subsection{Bellman Equation}\label{subsec:bellman_equation_sum_power_sum_outage_no_backtracking}

Let us denote the optimal expected cost-to-go at state $(r,w)$ and at 
state $\mathbf{0}$ be $J(r,w)$ and $J(\mathbf{0})$ respectively. 
Note that here we have an infinite horizon
total cost MDP with a finite state space and finite action
space. The assumption P of Chapter~$3$ in \cite{bertsekas07dynamic-programming-optimal-control-2} is satisfied
here, since the single-stage costs are nonnegative (power, outage and relay costs are all nonnegative). 
Hence, by the theory developed in \cite{bertsekas07dynamic-programming-optimal-control-2}, 
we can restrict ourselves to the class of stationary deterministic
Markov policies. Any deterministic Markov policy $\pi$ is a sequence $\{\mu_k\}_{k \geq 1}$ of mappings from 
the state space to the action space. A deterministic Markov policy is called ``stationary'' if $\mu_k=\mu$ for all $k \geq 1$. 

By Proposition~$3.1.1$ of \cite{bertsekas07dynamic-programming-optimal-control-2}, the optimal value function $J(\cdot)$
satisfies the Bellman equation which is given by, for all $(A+1) \leq r \leq (A+B-1)$,

\footnotesize
\begin{eqnarray}
 J(r,w)&=&\min \bigg\{ \min_{\gamma \in \mathcal{S}} ( \gamma+\xi_o P_{out}(r,\gamma,w) )+\xi_r + J(\mathbf{0}), \nonumber\\
&& \theta \mathbb{E}_W \min_{\gamma \in \mathcal{S}} (\gamma+ \xi_o  P_{out} (r+1,\gamma,W)) \nonumber\\
&& + (1-\theta)\mathbb{E}_W J(r+1,W)  \bigg\},\nonumber\\
J(A+B, w)&=&\min_{\gamma \in \mathcal{S}} (\xi_r+\gamma+\xi_o P_{out}(A+B,\gamma,w) ) + J(\mathbf{0}) \nonumber\\
J(\mathbf{0})&=& \sum_{k=1}^{A+1} (1-\theta)^{k-1}\theta \mathbb{E}_W \min_{\gamma \in \mathcal{S}} (\gamma+ \xi_o P_{out} (k,\gamma,W))\nonumber\\
&& +(1-\theta)^{A+1}\mathbb{E}_W J(A+1,W)\label{eqn:bellman_equation_sum_power_sum_outage_no_backtracking}
\end{eqnarray}
\normalsize

The equation for $J(r,w)$ can be understood as follows. 
If the current state is $(r,w), (A+1) \leq r \leq (A+B-1)$ and the line has not ended yet, we can either place a relay 
and use some $\gamma$ power in it, or we may not place. If we place, 
a cost $\min_{\gamma \in \mathcal{S}} (\xi_r+ \gamma+\xi_o P_{out}(r,\gamma,w) )$ is incurred at the current step and the 
cost-to-go from there is $J(\mathbf{0})$ since the decision process regenerates at the point. If we do not 
place a relay, the line will end with probability $\theta$ in the next step, in which case a 
cost $\mathbb{E}_W \min_{\gamma \in \mathcal{S}} (\gamma+ \xi_o P_{out} (r+1,\gamma,W))$ will be 
incurred. If the line does not end in the next step, the next state will be a random state $(r+1,W)$ 
and a mean cost of 
$\mathbb{E}_W J(r+1,W)$ will be incurred. At state $(A+B,w)$ the only possible decision is to place a relay; hence the 
expression follows. 
At state $\mathbf{0}$, the deployment agent starts walking until he encounters the source location or location $(A+1)$; if the line 
ends at step $k, 1 \leq k \leq A+1$ (with probability $(1-\theta)^{k-1}\theta$), a cost of  
$\mathbb{E}_W \min_{\gamma \in \mathcal{S}} (\gamma+ \xi_o P_{out} (k,\gamma,W))$ is incurred. If the line does not end within 
$(A+1)$ steps (this event has probability $(1-\theta)^{A+1}$), the next state will be a random state $(A+1,W)$.

\subsection{Value Iteration}\label{subsec:value_iteration_sum_power_sum_outage_no_backtracking}

The value iteration for (\ref{eqn:sum_power_discounted_no_backtracking}) is given by, for all $(A+1) \leq r \leq (A+B-1)$:

\footnotesize
\begin{eqnarray}
 J^{(k+1)}(r,w)&=&\min \bigg\{ \min_{\gamma \in \mathcal{S}} ( \gamma+\xi_o P_{out}(r,\gamma,w) ) +\xi_r+ J^{(k)}(\mathbf{0}), \nonumber\\
&& \theta \mathbb{E}_W \min_{\gamma \in \mathcal{S}} (\gamma+ \xi_o  P_{out} (r+1,\gamma,W)) \nonumber\\
&& + (1-\theta)\mathbb{E}_W J^{(k)}(r+1,W)  \bigg\},\nonumber\\
J^{(k+1)}(A+B, w)&=&\min_{\gamma \in \mathcal{S}} (\gamma+\xi_o P_{out}(A+B,\gamma,w)+\xi_r ) \nonumber\\
&& + J^{(k)}(\mathbf{0}) \nonumber\\
J^{(k+1)}(\mathbf{0})&=& \sum_{k=1}^{A+1} (1-\theta)^{k-1}\theta \mathbb{E}_W \min_{\gamma \in \mathcal{S}} (\gamma+ \xi_o P_{out} (k,\gamma,W))\nonumber\\
&& +(1-\theta)^{A+1}\mathbb{E}_W J^{(k)}(A+1,W)\label{eqn:value_iteration_sum_power_sum_outage_no_backtracking}
\end{eqnarray}
\normalsize
with $J^{(0)}(r,w):=0$ for all $r,w$ and $J^{(0)}(\mathbf{0}):=0$. 

\begin{lem}\label{lemma:value_iteration_sum_power_sum_outage_no_backtracking}
 The value iteration (\ref{eqn:value_iteration_sum_power_sum_outage_no_backtracking}) provides a nondecreasing
sequence of iterates that converges to the optimal value function, i.e., $J^{(k)}(r, w) \uparrow J(r, w)$ for all $r,w$ and 
$J^{(k)}(\mathbf{0}) \uparrow J(\mathbf{0})$.
\end{lem}

\begin{proof}
 See Appendix~\ref{appendix:sum_power_sum_outage_no_backtracking_discounted}.
\end{proof}

\subsection{Policy Structure}\label{subsec:policy_structure_sum_power_sum_outage_no_backtracking}

\begin{lem}\label{lemma:value_function_properties_sum_power_sum_outage_no_backtracking}
 $J(r,w)$ is increasing in $r$, $\xi_o$ and $\xi_r$, decreasing in $w$, and jointly concave in 
$\xi_o$ and $\xi_r$. $J(\mathbf{0})$ is increasing and jointly concave in $\xi_o$ and $\xi_r$.
\end{lem}

\begin{proof}
 See Appendix~\ref{appendix:sum_power_sum_outage_no_backtracking_discounted}.
\end{proof}

\begin{thm}\label{theorem:policy_structure_sum_power_sum_outage_no_backtracking}
 At state $(r,w)$ ($A+1 \leq r \leq A+B-1$), the optimal decision is to place a relay iff 
$\min_{\gamma \in \mathcal{S}} (\gamma+\xi_o P_{out}(r,\gamma,w) ) \leq c_{th}(r)$ where $c_{th}(r)$ is a 
threshold increasing in $r$. In this case if the decision is to place a relay, the optimal power to be selected is given by 
$\argmin_{\gamma \in \mathcal{S}} \bigg(\gamma+\xi_o P_{out}(r,\gamma,w)\bigg)$. At state $(A+B,w)$, the 
optimal power to be selected is $\argmin_{\gamma \in \mathcal{S}} \bigg(\gamma+\xi_o P_{out}(A+B,\gamma,w)\bigg)$.
\end{thm}

\begin{proof}
 See Appendix~\ref{appendix:sum_power_sum_outage_no_backtracking_discounted}.
\end{proof}

{\em Remark:} $c_{th}(r)$ captures the effect of the tradeoff that if we place relays far apart, the cost due to outage 
increases, but the cost of placing the relays decreases. $c_{th}(r)$ is increasing in $r$ because 
$P_{out}(r,\gamma,w)$ is increasing in $r$ for any $\gamma,w$.

Note that the threshold $c_{th}(r)$ does not depend on $w$, due to the fact that 
shadowing is i.i.d across links.\footnote{Though the length of the line is assumed to be geometrically 
distributed, similar approach as in this paper can be used to analyze the case where the length of the line is constant and known. 
The only difference will be that the optimal policy will be nonstationary.}

\subsection{Computation of the Optimal Policy}\label{subsec:policy_computation_sum_power_sum_outage_no_backtracking}

Let us write $V(r):=\mathbb{E}_W J \left(r, W \right)=\sum_{w \in \mathcal{W}} p_W(w) J \left(r, w \right)$, i.e., 
for all $r \in \{A+1,A+2,\cdots,A+B\}$, and  
$V(\mathbf{0}):=J(\mathbf{0})$. Also, for each stage $k \geq 0$ of the value iteration 
(\ref{eqn:value_iteration_sum_power_sum_outage_no_backtracking}), 
define $V^{(k)}(r):=\mathbb{E}_W J^{(k)}\left(r, W \right)$ and 
$V^{(k)}(\mathbf{0}):=J^{(k)}(\mathbf{0})$.

Observe that from the value iteration (\ref{eqn:value_iteration_sum_power_sum_outage_no_backtracking}), we obtain 
for all  $(A+1) \leq r \leq (A+B-1)$:

\footnotesize
\begin{eqnarray}
 V^{(k+1)}(r)&=&\sum_{w \in  \mathcal{W}} p_W(w) \min \bigg\{ \min_{\gamma \in \mathcal{S}} \bigg( \gamma + \nonumber\\
&& \xi_o P_{out}(r,\gamma,w) +\xi_r \bigg) + V^{(k)}(\mathbf{0}), \nonumber\\
&& \theta \mathbb{E}_W \min_{\gamma \in \mathcal{S}} (\gamma+ \xi_o  P_{out} (r+1,\gamma,W)) \nonumber\\
&& + (1-\theta) V^{(k)}(r+1)  \bigg\},  \nonumber\\
V^{(k+1)}(A+B)&=& \sum_{w \in \mathcal{W}} p_W(w) \min_{\gamma \in \mathcal{S}} \bigg(\gamma + \nonumber\\
&& +\xi_o P_{out}(A+B,\gamma,w) +\xi_r \bigg) + V^{(k)}(\mathbf{0}) \nonumber\\
V^{(k+1)}(\mathbf{0})&=& \sum_{k=1}^{A+1} (1-\theta)^{k-1} \theta \mathbb{E}_W \min_{\gamma \in \mathcal{S}} \bigg(\gamma+ \xi_o P_{out} (k,\gamma,W)\bigg)\nonumber\\
&& +(1-\theta)^{A+1}V^{(k)}(A+1)\label{eqn:function_iteration_sum_power_sum_outage_no_backtracking}
\end{eqnarray}
\normalsize
with $V^{(0)}(r):=0$ for all $A+1 \leq r \leq A+B$ and $V^{(0)}(\mathbf{0}):=0$. 

Since $J^{(k)}(r,w) \uparrow J(r,w)$ for each $r$, $w$ and 
$J^{(k)}(\mathbf{0}) \uparrow J(\mathbf{0})$ as $k \uparrow \infty$, we can argue that 
$V^{(k)}(r)\uparrow \mathbb{E}_W J(r, W)$ for all 
$r$ (by Monotone Convergence Theorem) and  
$V^{(k)}(\mathbf{0})\uparrow J(\mathbf{0})$. Thus, 
$V^{(k)}(r)\uparrow V(r)$ and $V^{(k)}(\mathbf{0}) \uparrow V(\mathbf{0})$. 
 Hence, by the function iteration (\ref{eqn:function_iteration_sum_power_sum_outage_no_backtracking}), we obtain $V(\mathbf{0})$ and 
$V(r)$ for all $r \geq 1$. Then, from (\ref{eqn:c_th_r_expression}), 
we can compute $c_{th}(r)$. Thus, for this iteration, we need not keep track of the cost-to-go values 
$J^{(k)}(r, w)$ for each state $(r,w)$, at each stage $k$; we simply need to keep track of $V^{(k)}(\mathbf{0})$ and 
$V^{(k)}(r)$ for each $r$.

\begin{figure*}[t!]
\footnotesize
 \begin{eqnarray}
 J(r,w,\gamma_{max})
&=&\min \bigg\{ \min_{\gamma \in \mathcal{S}} \bigg( \xi_o P_{out}(r,\gamma,w)+\xi_r 
+ J(\mathbf{0};\max\{\gamma,\gamma_{max}\}) \bigg), 
 \theta \mathbb{E}_W \min_{\gamma \in \mathcal{S}} \bigg(\max\{\gamma,\gamma_{max}\} + \xi_o  P_{out} (r+1,\gamma,W)\bigg) \nonumber\\
&& + (1-\theta)\mathbb{E}_W J(r+1,W,\gamma_{max})  \bigg\}, \, \forall (A+1) \leq r \leq (A+B-1) \nonumber\\
J(A+B, w, \gamma_{max})
&=&\min_{\gamma \in \mathcal{S}} \bigg(\xi_o P_{out}(A+B,\gamma,w)+\xi_r 
 + J(\mathbf{0};\max\{\gamma,\gamma_{max}\})\bigg) \nonumber\\
 J(\mathbf{0};\gamma_{max})
&=& \sum_{k=1}^{A+1}(1-\theta)^{k-1}\theta \mathbb{E}_W \min_{\gamma \in \mathcal{S}} \bigg(\max\{\gamma,\gamma_{max}\}+ 
 \xi_o P_{out} (k,\gamma,W)\bigg) +(1-\theta)^{A+1}\mathbb{E}_W J(A+1,W,\gamma_{max}) \label{eqn:bellman_equation_max_power_sum_outage_no_backtracking}
\end{eqnarray}
\normalsize
\hrule
\end{figure*}

\section{Impromptu Deployment for Geometrically Distributed Length without Backtracking: Max Power and 
Sum Outage Objective}\label{sec:max_power_sum_outage_no_backtracking_discounted}

\begin{figure*}[t!]
\footnotesize
 \begin{eqnarray}
 J^{(k+1)}(r,w,\gamma_{max})
&=&\min \bigg\{ \min_{\gamma \in \mathcal{S}} \bigg( \xi_o P_{out}(r,\gamma,w)+\xi_r 
+ J^{(k)}(\mathbf{0};\max\{\gamma,\gamma_{max}\}) \bigg), 
 \theta \mathbb{E}_W \min_{\gamma \in \mathcal{S}} \bigg(\max\{\gamma,\gamma_{max}\}  \nonumber\\
&& + \xi_o  P_{out} (r+1,\gamma,W)\bigg) + (1-\theta)\mathbb{E}_W J^{(k)}(r+1,W,\gamma_{max})  \bigg\}, \, \forall (A+1) \leq r \leq (A+B-1) \nonumber\\
J^{(k+1)}(A+B, w, \gamma_{max})
&=&\min_{\gamma \in \mathcal{S}} \bigg(\xi_o P_{out}(A+B,\gamma,w)+\xi_r 
 + J^{(k)}(\mathbf{0};\max\{\gamma,\gamma_{max}\})\bigg) \nonumber\\
J^{(k+1)}(\mathbf{0};\gamma_{max})
&=& \sum_{k=1}^{A+1}(1-\theta)^{k-1}\theta \mathbb{E}_W \min_{\gamma \in \mathcal{S}} \bigg(\max\{\gamma,\gamma_{max}\}+ 
 \xi_o P_{out} (k,\gamma,W)\bigg) +(1-\theta)^{A+1}\mathbb{E}_W J^{(k)}(A+1,W,\gamma_{max})\label{eqn:value_iteration_max_power_sum_outage_no_backtracking}
\end{eqnarray}
\normalsize
\hrule
\end{figure*}

\begin{figure*}[t!]
\small
 \begin{eqnarray}
 V^{(k+1)}(r,\gamma_{max})
&=& \sum_{w \in \mathcal{W}} p_{W}(w) \min \bigg\{ \min_{\gamma \in \mathcal{S}} \bigg( \xi_o P_{out}(r,\gamma,w)+\xi_r 
+ V^{(k)}(\mathbf{0};\max\{\gamma,\gamma_{max}\}) \bigg), 
 \theta \mathbb{E}_W \min_{\gamma \in \mathcal{S}} \bigg(\max\{\gamma,\gamma_{max}\}  \nonumber\\
&& + \xi_o  P_{out} (r+1,\gamma,W)\bigg) + (1-\theta)V^{(k)}(r+1,\gamma_{max})  \bigg\}, \, \forall (A+1) \leq r \leq (A+B-1) \nonumber\\
V^{(k+1)}(A+B, \gamma_{max})
&=& \sum_{w \in \mathcal{W}} p_{W}(w) \min_{\gamma \in \mathcal{S}} \bigg(\xi_o P_{out}(A+B,\gamma,w)+\xi_r 
 + V^{(k)}(\mathbf{0};\max\{\gamma,\gamma_{max}\})\bigg) \nonumber\\
 V^{(k+1)}(\mathbf{0};\gamma_{max})
&=& \sum_{k=1}^{A+1}(1-\theta)^{k-1}\theta \mathbb{E}_W \min_{\gamma \in \mathcal{S}} \bigg(\max\{\gamma,\gamma_{max}\}+ 
 \xi_o P_{out} (k,\gamma,W)\bigg) +(1-\theta)^{A+1}V^{(k)}(A+1,\gamma_{max}) \label{eqn:function_iteration_max_power_sum_outage_no_backtracking}
\end{eqnarray}
\hrule
\end{figure*}

\subsection{Problem Formulation}\label{subsec:problem_formulation_max_power_sum_outage_no_backtracking}

Here we seek to solve the following problem without backtracking, for a line having geometrically distributed length:

\begin{equation}
 \min_{\pi \in \Pi} \mathbb{E}_{\pi} \bigg( \max_{i \in \{1,2,\cdots,N+1\}} \Gamma^{(i,i-1)}+ 
\xi_o \sum_{i=1}^{N+1}P_{out}^{(i,i-1)} +\xi_r N \bigg)
\label{eqn:max_power_discounted_no_backtracking}
\end{equation}

We formulate (\ref{eqn:max_power_discounted_no_backtracking}) as an MDP 
with $(r,w, \gamma_{max})$ as a typical state, where $\gamma_{max}$ is the maximum transmit power 
used by already deployed nodes. 
At state $(r,w, \gamma_{max}), (A+1) \leq r \leq (A+B-1), w \in \mathcal{W}$, 
the action is either to 
place a relay and select some transmit power $\gamma \in \mathcal{S}$, or not to place. When $r=A+B$, we must place a relay. 
The state of the system at 
a point where a relay has just been placed and the maximum power used in all previous links is $\gamma_{\max}$, 
is denoted by $(\mathbf{0};\gamma_{max})$. The state at the sink is $(\mathbf{0};\gamma_{max})$ with $\gamma_{max}=0$. 
Hence, in our current problem formulation, $\gamma_{max}$ can take values from the set $\{0\} \cup \mathcal{S}$. 
At state $(\mathbf{0};\gamma_{max})$, the only possible action is to move to the next step. 
When the source is placed, the process terminates.

\subsection{Bellman Equation}\label{subsec:bellman_equation_max_power_sum_outage_no_backtracking}

Unlike problem (\ref{eqn:sum_power_discounted_no_backtracking}), here the cost of the maximum power 
over all links is incurred when the source is placed. However, the outage and relay costs are incurred whenever a node is placed. 

The optimal value function $J(\cdot)$
satisfies the Bellman equation given by (\ref{eqn:bellman_equation_max_power_sum_outage_no_backtracking}). 
This equation  can be understood as follows. 
If the current state is $(r,w,\gamma_{max}), (A+1) \leq r \leq (A+B-1)$ and the line has not ended yet, 
we can either place a relay 
and use some $\gamma$ power in it, or we may not place. If we place and use power $\gamma$, 
a cost $(\xi_o P_{out}(r,\gamma,w)+\xi_r )$ is incurred at the current step and the 
state becomes $(\mathbf{0};\max\{\gamma,\gamma_{max}\})$. If we do not 
place a relay, the line will end with probability $\theta$ in the next step, in which case a 
cost $\mathbb{E}_W \min_{\gamma \in \mathcal{S}} \bigg(\max\{\gamma,\gamma_{max}\} + \xi_o  P_{out} (r+1,\gamma,W)\bigg)$ will be 
incurred. If the line does not end in the next step, the next state will be a random state $(r+1,W, \gamma_{max})$ 
and a mean cost of 
$\mathbb{E}_W J(r+1,W,\gamma_{max})$ will be incurred. At state $(A+B,w,\gamma_{max})$ 
the only possible decision is to place a relay; hence the 
expression follows. 
At state $(\mathbf{0};\gamma_{max})$, the deployment agent explores at least upto the $(A+1)$-st step. If the line ends at a distance 
of $k$-th step ($1 \leq k \leq A+1$) (with probability $(1-\theta)^{k-1}\theta$), 
a cost $\mathbb{E}_W \min_{\gamma \in \mathcal{S}} \bigg(\max\{\gamma,\gamma_{max}\}+ 
 \xi_o P_{out} (k,\gamma,W)\bigg)$ is incurred. If the line does not end in 
$(A+1)$ steps (with probability $(1-\theta)^{A+1}$), the next state will be a random state $(A+1,W,\gamma_{max})$.

\subsection{Value Iteration}\label{subsec:value_iteration_max_power_sum_outage_no_backtracking}

Starting with $J^{(0)}(r,w,\gamma_{max})=0$ and $J^{(0)}(\mathbf{0};\gamma_{max})=0$ for all $r,w,\gamma_{max}$, 
the value iteration for problem (\ref{eqn:max_power_discounted_no_backtracking}) is given by 
(\ref{eqn:value_iteration_max_power_sum_outage_no_backtracking}).

\begin{lem}\label{lemma:value_iteration_max_power_sum_outage_no_backtracking}
 The value iteration (\ref{eqn:value_iteration_max_power_sum_outage_no_backtracking}) provides a nondecreasing
sequence of iterates that converges to the optimal value function, i.e., 
$J^{(k)}(r, w, \gamma_{max}) \uparrow J(r, w, \gamma_{max})$ 
and $J^{(k)}(\mathbf{0};\gamma_{max}) \uparrow J(\mathbf{0};\gamma_{max})$.
\end{lem}

\begin{proof}
 Proof follows along the same line of arguments as in Lemma~\ref{lemma:value_iteration_sum_power_sum_outage_no_backtracking}.
\end{proof}

\subsection{Policy Structure}\label{subsec:policy_structure_max_power_sum_outage_no_backtracking}

\begin{lem}\label{lemma:value_function_properties_max_power_sum_outage_no_backtracking}
 $J(r,w,\gamma_{max})$ is increasing in $r$, $\gamma_{max}$, $\xi_o$ and $\xi_r$, decreasing in $w$, and jointly concave in 
$\xi_o$ and $\xi_r$. $J(\mathbf{0};\gamma_{max})$ is increasing and jointly concave in $\xi_o$ and $\xi_r$, and 
increasing in $\gamma_{max}$.
\end{lem}

\begin{proof}
 See Appendix~\ref{appendix:max_power_sum_outage_no_backtracking_discounted}.
\end{proof}

\begin{thm}\label{theorem:policy_structure_max_power_sum_outage_no_backtracking}
 At state $(r,w,\gamma_{max})$ ($A+1 \leq r \leq A+B-1$), the optimal decision is to place a relay iff 
$\min_{\gamma \in \mathcal{S}} \bigg( \xi_o P_{out}(r,\gamma,w)
+ J(\mathbf{0};\max\{\gamma,\gamma_{max}\}) \bigg) \leq c_{th}(r,\gamma_{max})$ where $c_{th}(r,\gamma_{max})$ 
is a threshold function increasing in $r$ and $\gamma_{max}$. 
A relay must be placed at $r=A+B$. If the decision is to place a relay, then the optimal transmit power 
for the new relay is given by $\argmin_{\gamma \in \mathcal{S}} \bigg( \xi_o P_{out}(r,\gamma,w)+\xi_r 
+ J(\mathbf{0};\max\{\gamma,\gamma_{max}\}) \bigg)$.
\end{thm}

\begin{proof}
 See Appendix~\ref{appendix:max_power_sum_outage_no_backtracking_discounted}.
\end{proof}

\subsection{Computation of the Optimal Policy}\label{subsec:policy_computation_max_power_sum_outage_no_backtracking}

Let us write 
$V(r,\gamma_{max}):=\mathbb{E}_W J \left(r, W,\gamma_{max} \right)=\sum_{w \in \mathcal{W}} p_W(w) J \left(r, w, \gamma_{max} \right)$, 
for all $r \in \{A+1,A+2,\cdots,A+B\}$ and all $\gamma_{max} \in \{0\} \cup \mathcal{S}$, 
$V(\mathbf{0};\gamma_{max}):=J(\mathbf{0};\gamma_{max})$. Also, for each stage $k \geq 0$ of the value iteration 
(\ref{eqn:value_iteration_max_power_sum_outage_no_backtracking}), 
define $V^{(k)}(r,\gamma_{max}):=\mathbb{E}_W J^{(k)} \left(r, W,\gamma_{max} \right)$, 
and $V^{(k)}(\mathbf{0};\gamma_{max}):=J^{(k)}(\mathbf{0};\gamma_{max})$. 

Observe that from the value iteration (\ref{eqn:value_iteration_max_power_sum_outage_no_backtracking}), we obtain 
(\ref{eqn:function_iteration_max_power_sum_outage_no_backtracking}). Using similar arguments as in 
Section~\ref{subsec:policy_computation_sum_power_sum_outage_no_backtracking}, we can conclude that 
$V^{(k)}(\cdot) \uparrow V(\cdot)$ in (\ref{eqn:function_iteration_max_power_sum_outage_no_backtracking}). Then for each $r,\gamma$, 
the value of $c_{th}(r,\gamma)$ can be computed from the function $V(\cdot)$, using the Bellman equation 
(\ref{eqn:bellman_equation_max_power_sum_outage_no_backtracking}).

\section{Impromptu Deployment for Geometrically Distributed Length with Backtracking: Sum Power and 
Sum Outage Objective}\label{sec:sum_power_sum_outage_with_backtracking_discounted}

\subsection{Problem Formulation}\label{subsec:problem_formulation_sum_power_sum_outage_with_backtracking_discounted}

Consider the deployment procedure as in Section~\ref{subsection:deployment_process_notation}, with the objective 
(\ref{eqn:sum_power_discounted_no_backtracking}), under the scenario where the length of the line 
is geometrically distributed with parameter $\theta$. We formulate this problem as an MDP with state space 
$\mathcal{W}^B \cup \{z;\mathbf{0}\}_{0 \leq z \leq B-1 }$. The deployment agent starts walking 
from the previous node location, explores the next $(A+B)$ steps and measures 
$\underline{w}=(w_{A+1},w_{A+2},\cdots,w_{A+B})$ which belongs to $\mathcal{W}^B$. 
The state $(z;\mathbf{0})$ means that a relay has already been placed at the current position and the residual length of the line 
from the current location is $(z+L_1)$ where $L_1 \sim Geometric(\theta)$. At state $\underline{w}$ an action 
$(u,\gamma)$ is taken, where $u \in \{A+1,A+2,\cdots,A+B\}$ and $\gamma \in \mathcal{S}$. 
At state $(z;\mathbf{0})$ the action is to explore next $(A+B)$ steps, out of which $B$ steps will involve measurements. Note that 
the link qualities obtained from these new measurements will be independent from the previous measurements, 
since here new links (transmitter-receiver pairs) are being measured.  
The state $(z;\mathbf{0})$ is needed for the following reason: suppose that at some state $\underline{w}$ the 
optimal decision is to place the next relay $u$ steps away from the previous relay, where $A+1 \leq u \leq A+B$. 
After placing this relay, the residual length of the line becomes 
$(A+B-u+L_1)$ where $L_1 \sim Geometric(\theta)$; the problem does not restart after the placement 
of a relay as it did in Section~\ref{sec:sum_power_sum_outage_no_backtracking_discounted}. 
When the line ends, the process terminates.

\subsection{Bellman Equation}\label{subsec:bellman_equation_sum_power_sum_outage_with_backtracking_discounted}

Following the same arguments as in Section~\ref{subsec:bellman_equation_sum_power_sum_outage_no_backtracking}, we can 
argue that the optimal expected cost-to-go function $J(\cdot)$ satisfies the following Bellman equation: 

\begin{eqnarray}
J(\underline{w})&=&\min_{u \in \{A+1,\cdots,A+B\},\gamma \in \mathcal{S}} \bigg\{ \xi_r+\gamma+ \nonumber\\
&& \xi_o P_{out}(u,\gamma,w_u)+J(A+B-u;\mathbf{0}) \bigg\}\nonumber\\
J(z;\mathbf{0})&=&\sum_{k=1}^{A+B-z} (1-\theta)^{k-1} \theta \mathbb{E}_{W_{z+k}}  \nonumber\\
&& \min_{\gamma \in \mathcal{S}} \bigg( \gamma+ \xi_o P_{out}(z+k,\gamma,W_{z+k}) \bigg) \nonumber\\
&& +(1-\theta)^{A+B-z} \sum_{\underline{w}} g(\underline{w}) J(\underline{w})
\label{eqn:bellman_equation_sum_power_sum_outage_with_backtracking}
\end{eqnarray}

When the state is $\underline{w}$, if the action $(u,\gamma)$ is taken then a cost of $\xi_r+\gamma+\xi_o P_{out}(u,\gamma,w_u)$ 
is incurred in the current step and the next state becomes $(A+B-u;\mathbf{0})$, resulting in an additional cost 
$J(A+B-u;\mathbf{0})$. If the state is $(z;\mathbf{0})$, the source can appear in the $(z+k)$-th step ($1 \leq k \leq A+B-z$) 
from the current location with probability $(1-\theta)^{k-1} \theta$ 
(since the residual length of the line is $z$ plus
a geometrically distributed random variable), in which case 
a mean cost of $\mathbb{E}_{W_{z+k}} \min_{\gamma \in \mathcal{S}} (\gamma+ \xi_o P_{out}(z+k,\gamma,W_{z+k}))$ is incurred 
in the last hop. 
If the line does not end in next $(A+B)$ steps (which has probability $(1-\theta)^{A+B-z}$), the next state becomes 
$\underline{w}$ with probability $g(\underline{w}):=\Pi_{r=A+1}^{A+B} p_{W_r}(w_r)$ (since shadowing is i.i.d across links). 
Note that the optimal expected cost-to-go at the sink node is $J(0;\mathbf{0})$.

\subsection{Value Iteration}\label{subsec:value_iteration_sum_power_sum_outage_with_backtracking_discounted}

The value iteration for this problem is given by:

 \begin{eqnarray}
J^{(k+1)}(\underline{w})&=&\min_{u \in \{A+1,\cdots,A+B\},\gamma \in \mathcal{S}} \bigg\{ \xi_r+\gamma+ \nonumber\\
&& \xi_o P_{out}(u,\gamma,w_u)+J^{(k)}(A+B-u;\mathbf{0}) \bigg\}\nonumber\\
J^{(k+1)}(z;\mathbf{0})&=&\sum_{k=1}^{A+B-z} (1-\theta)^{k-1} \theta \mathbb{E}_{W_{z+k}}  \nonumber\\
&& \min_{\gamma \in \mathcal{S}} \bigg( \gamma+ \xi_o P_{out}(z+k,\gamma,W_{z+k}) \bigg) \nonumber\\
&& +(1-\theta)^{A+B-z} \sum_{\underline{w}} g(\underline{w}) J^{(k)}(\underline{w})
\label{eqn:value_iteration_sum_power_sum_outage_with_backtracking}
\end{eqnarray}
with $J^{(0)}(\cdot):=0$ for all states.

\begin{lem}\label{lemma:properties_J_vs_xi_sum_power_sum_outage_with_backtracking}
 Each of $J(\underline{w})$ and $\{ J(z;\mathbf{0}) \}_{0 \leq z \leq B-1}$ is increasing and jointly concave  in $\xi_r$, $\xi_o$. 
\end{lem}

\begin{proof}
 The proof follows from the convergence of value iterates to the optimal value function, along the same 
lines as in Lemma~\ref{lemma:value_function_properties_sum_power_sum_outage_no_backtracking}.
\end{proof}

\begin{lem}\label{lemma:properties_J_vs_w_sum_power_sum_outage_with_backtracking}
 $J(\underline{w})$ is decreasing in each component of $\underline{w}$.
\end{lem}

\begin{proof}
 Note that for each $(u,\gamma)$, $P_{out}(u,\gamma,w_u)$ is decreasing in $w_u$. Hence, the result follows 
from the first equation in (\ref{eqn:bellman_equation_sum_power_sum_outage_with_backtracking}).
\end{proof}

\begin{lem}\label{lemma:properties_J_vs_z_sum_power_sum_outage_with_backtracking}
 $J(z;\mathbf{0})$ is increasing in $z$.
\end{lem}

\begin{proof}
See Appendix~\ref{appendix:sum_power_sum_outage_with_backtracking_discounted}.
\end{proof}

\subsection{Policy Structure}\label{subsec:policy_structure_sum_power_sum_outage_with_backtracking_discounted}
   
\begin{thm}\label{thm:policy_structure_sum_power_sum_outage_with_backtracking}
 The optimal action at state $\underline{w}$ is the pair $(u,\gamma)$ achieving the minimum in 
(\ref{eqn:bellman_equation_sum_power_sum_outage_with_backtracking}). The minimum is always achieved since we have finite action space.
\end{thm}

\subsection{Policy Computation}\label{subsec:policy_computation_sum_power_sum_outage_with_backtracking_discounted}

Note that in the $k$-th iteration of the value iteration obtained from the Bellman equation 
(\ref{eqn:bellman_equation_sum_power_sum_outage_with_backtracking}), we need to update $J^{(k)}(\underline{w})$ 
for $|\mathcal{W}|^{B}$ possible values of the state $\underline{w}$, which could be computationally very much expensive for 
large values of $|\mathcal{W}|$. Let us define the sequence $\{V^{(k)}\}_{k \geq 0}$ by 
$V^{(0)}=0$, $V^{(k)} =\sum_{\underline{w} \in \mathcal{W}^B} g(\underline{w}) J^{(k)}(\underline{w})$. Now consider the following 
iteration (with $J^{(0)}(z;\mathbf{0}):=0$ for all $0 \leq z \leq B-1$) obtained from 
(\ref{eqn:value_iteration_sum_power_sum_outage_with_backtracking}):

\begin{eqnarray}
V^{(k+1)}&=& \sum_{\underline{w}}g(\underline{w}) \min_{u \in \{A+1,\cdots,A+B\},\gamma \in \mathcal{S}} \bigg\{ \xi_r+\gamma+ \nonumber\\
&& \xi_o P_{out}(u,\gamma,w_u)+J^{(k)}(A+B-u;\mathbf{0}) \bigg\}\nonumber\\
J^{(k+1)}(z;\mathbf{0})&=&\sum_{k=1}^{A+B-z} (1-\theta)^{k-1} \theta \mathbb{E}_{W_{z+k}}  \nonumber\\
&& \min_{\gamma \in \mathcal{S}} \bigg( \gamma+ \xi_o P_{out}(z+k,\gamma,W_{z+k}) \bigg) \nonumber\\
&& +(1-\theta)^{A+B-z}  V^{(k)}, 0 \leq z \leq (B-1) \nonumber\\
\label{eqn:function_iteration_sum_power_sum_outage_with_backtracking}
\end{eqnarray}

By Monotone Convergence Theorem, $V^{(k)} \uparrow \sum_{\underline{w}} g(\underline{w}) J(\underline{w})$. Hence, 
we can just use the function iteration (\ref{eqn:function_iteration_sum_power_sum_outage_with_backtracking}) to compute the optimal 
value function, from which the policy can be computed. The advantage of this function iteration is that we need not update 
$J^{(k)}(\cdot)$ for each state.

\section{Impromptu Deployment for Geometrically Distributed Length with Backtracking: Max Power and 
Sum Outage Objective}\label{sec:max_power_sum_outage_with_backtracking_discounted}

\subsection{Problem Formulation}\label{subsec:problem_formulation_max_power_sum_outage_with_backtracking_discounted}

In this section, we seek to develop optimal placement policy with backtracking 
for the problem (\ref{eqn:max_power_discounted_no_backtracking}). We formulate this problem as an MDP with state space 
$\bigg(\mathcal{W}^B \cup \{z;\mathbf{0}\}_{0 \leq z \leq B-1 } \bigg) \times (\mathcal{S} \cup \{0\})$. 
The state $(z;\mathbf{0}; \gamma_{max})$ ($\gamma_{max} \in \mathcal{S} \cup \{0\}$) 
means that a relay has already been placed at 
the current position, the residual length of the line 
from the current location is $(z+L_1)$ where $L_1 \sim Geometric(\theta)$, and the maximum transmit power 
used so far by the previous nodes is $\gamma_{max}$. At state $(\underline{w}; \gamma_{max})$ an action 
$(u,\gamma)$ is taken, where $u \in \{A+1,A+2,\cdots,A+B\}$ and $\gamma \in \mathcal{S}$. 
At state $(z;\mathbf{0}; \gamma_{max})$ the action is to explore next $(A+B)$ steps.

\subsection{Bellman Equation}\label{subsec:bellman_equation_max_power_sum_outage_with_backtracking_discounted}

The optimal expected cost-to-go function $J(\cdot)$ satisfies the following Bellman equation: 

\footnotesize
\begin{eqnarray}
J(\underline{w}; \gamma_{max})&=&\min_{u \in \{A+1,\cdots,A+B\},\gamma \in \mathcal{S}} \bigg\{ \xi_o P_{out}(u,\gamma,w_u) \nonumber\\
&& + \xi_r +J(A+B-u;\mathbf{0}; \max\{\gamma,\gamma_{max}\}) \bigg\}\nonumber\\
J(z;\mathbf{0};\gamma_{max})&=&\sum_{k=1}^{A+B-z} (1-\theta)^{k-1} \theta \mathbb{E}_{W_{z+k}}\min_{\gamma \in \mathcal{S}} \bigg( \max\{\gamma,\gamma_{max}\}  \nonumber\\
&& + \xi_o P_{out}(z+k,\gamma,W_{z+k}) \bigg) \nonumber\\
&& +(1-\theta)^{A+B-z} \sum_{\underline{w}} g(\underline{w}) J(\underline{w};\gamma_{max})
\label{eqn:bellman_equation_max_power_sum_outage_with_backtracking}
\end{eqnarray}
\normalsize
When the state is $(\underline{w};\gamma_{max})$, if the action 
$(u,\gamma)$ is taken then a cost of $\xi_o P_{out}(u,\gamma,w_u)+\xi_r$ 
is incurred in the current step and the next state becomes $(A+B-u;\mathbf{0};\gamma_{max})$, resulting in an additional cost 
$J(A+B-u;\mathbf{0};\gamma_{max})$. If the state is $(z;\mathbf{0};\gamma_{max})$, 
the source can appear in the $(z+k)$-th step ($1 \leq k \leq A+B-z$) 
from the current location with probability $(1-\theta)^{k-1} \theta$ 
(since the residual length of the line is $z$ plus a geometrically distributed random variable), 
resulting in  
a mean cost of $\mathbb{E}_{W_{z+k}} \min_{\gamma \in \mathcal{S}} (\max \{\gamma,\gamma_{max}\}+ \xi_o P_{out}(z+k,\gamma,W_{z+k}))$, 
which is a combination of the max power used in the network and the outage probability of the last hop. 
If the line does not end in next $(A+B)$ steps (which has probability $(1-\theta)^{A+B-z}$), the next state becomes 
$(\underline{w};\gamma_{max})$ with probability 
$g(\underline{w}):=\Pi_{r=A+1}^{A+B} p_{W_r}(w_r)$ (since shadowing is i.i.d across links).

\subsection{Value Iteration}\label{subsec:value_iteration_max_power_sum_outage_with_backtracking_discounted}

The value iteration for this problem is given by:

\footnotesize
\begin{eqnarray}
&&J^{(k+1)}(\underline{w}; \gamma_{max}) = \min_{u \in \{A+1,\cdots,A+B\},\gamma \in \mathcal{S}} \bigg\{ \xi_o P_{out}(u,\gamma,w_u) \nonumber\\
&& + \xi_r +J^{(k)}(A+B-u;\mathbf{0}; \max\{\gamma,\gamma_{max}\}) \bigg\}\nonumber\\
&& J^{(k+1)}(z;\mathbf{0};\gamma_{max})=\sum_{k=1}^{A+B-z} (1-\theta)^{k-1} \theta \mathbb{E}_{W_{z+k}} \nonumber\\
&& \min_{\gamma \in \mathcal{S}} \bigg( \max\{\gamma,\gamma_{max}\} + \xi_o P_{out}(z+k,\gamma,W_{z+k}) \bigg) \nonumber\\
&& +(1-\theta)^{A+B-z} \sum_{\underline{w}} g(\underline{w}) J^{(k)}(\underline{w};\gamma_{max})
\label{eqn:value_iteration_max_power_sum_outage_with_backtracking}
\end{eqnarray}
\normalsize
with $J^{(0)}(\cdot):=0$ for all states.

\begin{lem}\label{lemma:properties_J_vs_xi_max_power_sum_outage_with_backtracking}
 Each of $J(\underline{w};\gamma_{max})$ and $\{ J(z;\mathbf{0};\gamma_{max}) \}_{0 \leq z \leq B-1}$ 
is increasing and jointly concave  in $\xi_r$, $\xi_o$, and increasing in $\gamma_{max}$. 
\end{lem}

\begin{proof}
 The proof follows from the convergence of value iterates to the optimal value function, along the same 
lines as in Lemma~\ref{lemma:value_function_properties_max_power_sum_outage_no_backtracking}.
\end{proof}

\begin{lem}\label{lemma:properties_J_vs_w_max_power_sum_outage_with_backtracking}
 $J(\underline{w};\gamma_{max})$ is decreasing in each component of $\underline{w}$.
\end{lem}

\begin{proof}
 Note that for each $(u,\gamma)$, $P_{out}(u,\gamma,w_u)$ is decreasing in $w_u$. Hence, the result follows 
from (\ref{eqn:bellman_equation_max_power_sum_outage_with_backtracking}).
\end{proof}

\begin{lem}\label{lemma:properties_J_vs_z_max_power_sum_outage_with_backtracking}
 $J(z;\mathbf{0};\gamma_{max})$ is increasing in $z$.
\end{lem}

\begin{proof}
It is easy to show that $J(z+1;\mathbf{0};\gamma_{max}) \geq J(z;\mathbf{0}; \gamma_{max})$, by similar arguments as in 
the proof of Lemma~\ref{lemma:properties_J_vs_z_max_power_sum_outage_with_backtracking}. 
\end{proof}

\subsection{Policy Structure}\label{subsec:policy_structure_max_power_sum_outage_with_backtracking_discounted}
   
\begin{thm}\label{thm:policy_structure_max_power_sum_outage_with_backtracking}
 The optimal action at state $(\underline{w};\gamma_{max})$ is the pair $(u,\gamma)$ achieving the minimum in 
(\ref{eqn:bellman_equation_max_power_sum_outage_with_backtracking}). 
The minimum is always achieved since we have finite action space.
\end{thm}

\subsection{Policy Computation}\label{subsec:policy_computation_max_power_sum_outage_with_backtracking_discounted}

Note that in the $k$-th iteration of the value iteration obtained from the Bellman equation 
(\ref{eqn:bellman_equation_max_power_sum_outage_with_backtracking}), we need to update $J^{(k)}(\underline{w};\gamma_{max})$ 
for $|\mathcal{W}|^{B}|(\mathcal{S}|+1)$ possible values of the state $(\underline{w};\gamma_{max})$, 
which could be computationally very much expensive for 
large values of $|\mathcal{W}|$. Let us define the sequence of functions $\{V^{(k)}(\gamma_{max})\}_{k \geq 0}$ by 
$V^{(0)}(\gamma_{max})=0$, $V^{(k)}(\gamma_{max}) =\sum_{\underline{w} \in \mathcal{W}^B} g(\underline{w}) J^{(k)}(\underline{w};\gamma_{max})$. 
Now consider the following 
iteration (with $J^{(0)}(z;\mathbf{0};\gamma_{max}):=0$ for all $0 \leq z \leq B-1$) obtained from 
(\ref{eqn:value_iteration_max_power_sum_outage_with_backtracking}):

\begin{eqnarray}
&&V^{(k+1)}(\gamma_{max})\nonumber\\
&& =\sum_{\underline{w}} g(\underline{w}) \min_{u \in \{A+1,\cdots,A+B\},\gamma \in \mathcal{S}} \bigg\{ \xi_o P_{out}(u,\gamma,w_u) \nonumber\\
&& + \xi_r +J^{(k)}(A+B-u;\mathbf{0}; \max\{\gamma,\gamma_{max}\}) \bigg\}\nonumber\\
&& J^{(k+1)}(z;\mathbf{0};\gamma_{max})=\sum_{k=1}^{A+B-z} (1-\theta)^{k-1} \theta \mathbb{E}_{W_{z+k}} \nonumber\\
&& \min_{\gamma \in \mathcal{S}} \bigg( \max\{\gamma,\gamma_{max}\} + \xi_o P_{out}(z+k,\gamma,W_{z+k}) \bigg) \nonumber\\
&& +(1-\theta)^{A+B-z} V^{(k)}(\gamma_{max})
\label{eqn:function_iteration_max_power_sum_outage_with_backtracking}
\end{eqnarray}

By Monotone Convergence Theorem, 
$V^{(k)} (\gamma_{max}) \uparrow \sum_{\underline{w}} g(\underline{w}) J(\underline{w};\gamma_{max})$. Hence, 
we can just use the function iteration (\ref{eqn:function_iteration_max_power_sum_outage_with_backtracking}) 
to compute the optimal 
value function, from which the policy can be computed. The advantage of this function iteration is that we need node update 
$J^{(k)}(\cdot)$ for each state.

\subsection{Comparison of the Optimal Expected Costs of Problem~(\ref{eqn:sum_power_discounted_no_backtracking}) and 
Problem~(\ref{eqn:cost_function_max_power_sum_outage})}\label{subsubsec:comparison_costs_max_power_vs_sum_power}

\begin{thm}\label{thm:comparison_costs_max_power_vs_sum_power}
 Under the same class of policies, the optimal expected cost for Problem~(\ref{eqn:sum_power_discounted_no_backtracking}) is always 
greater than or equal to that of Problem~(\ref{eqn:max_power_discounted_no_backtracking}).
\end{thm}
\begin{proof}
 Let $\Pi$ be a class of policies, and let $\pi \in \Pi$ be a specific policy. Consider any realization of $L$ and  
any realization of shadowing 
in all potential links; under policy $\pi$, relays will be placed at some locations and the relays and the source will use 
some transmit power levels. But, for any such deployed network the sum power is 
always greater than or equal to the max power, and hence we can write 
$J_{\pi}^{sum} \geq J_{\pi}^{max}$, where $J_{\pi}^{sum}$ and $J_{\pi}^{max}$ are the expected costs under policy $\pi$ of 
the problems (\ref{eqn:sum_power_discounted_no_backtracking}) and (\ref{eqn:max_power_discounted_no_backtracking}) respectively.  
Hence, $\inf_{\pi \in \Pi} J_{\pi}^{sum} \geq \inf_{\pi \in \Pi} J_{\pi}^{max}$, which completes the proof.
\end{proof}

\section{Average Cost Per Step: With and Without Backtracking}\label{sec:backtracking_average_cost}

Consider the deployment process as described in Section~\ref{subsection:deployment_process_notation}. 
After making the measurements $(w_{A+1},w_{A+2},\cdots,w_{A+B})$, the deployment agent 
chooses one integer $u$ from the set $\{A+1,A+2,\cdots,A+B\}$ 
and places the relay $u$ steps away from the last relay and also decides at what transmit 
power $\gamma \in \mathcal{S}$ the new relay should operate. 
The objective is to minimize the long-run expected average cost per step.

\subsection{Problem Formulation}\label{subsec:smdp-formulation}

We formulate our problem as a Semi-Markov Decision Process (SMDP) with state space $\mathcal{W}^B$ and action space 
$\{A+1,A+2,\cdots,A+B\} \times \mathcal{S}$. 
After placing a relay, the deployment agent measures $\underline{w}:=(w_{A+1},w_{A+2},\cdots,w_{A+B})$ 
which is the state in our SMDP. At state $\underline{w}$, if the action $(u,\gamma)$ is taken (where $A+1 \leq u \leq A+B$ 
and $\gamma \in \mathcal{S}$), the cost 
$c(\underline{w},u,\gamma):=(\gamma+\xi_o P_{out}(u,\gamma,w_u)+\xi_r)$ is incurred and the next 
state becomes $\underline{w}^{'}:=(w_{A+1}^{'},w_{A+2}^{'},\cdots,w_{A+B}^{'})$ with probability 
$g(\underline{w}^{'}):=\Pi_{r=A+1}^{A+B}p_{W_{r}}(w_{r}^{'})$ (since shadowing is i.i.d across links). 
A deterministic Markov policy $\pi$ is a sequence of mappings $\{\mu_k\}_{k \geq 1}$ 
from the state space to the action space, and it is called a stationary policy if $\mu_k=\mu$ for all $k \geq 1$. 
Let us denote, by the vector-valued random variable $\underline{W}(k)$, the state at the $k$-th decision instant, and by 
$\mu_k(\underline{W}(k))$ the action at the $k$-th decision instant. 
For a deterministic Markov policy $\{\mu_k\}_{k \geq 1}$, let us define the functions 
$\mu_k^{(1)}:\mathcal{W}^B \rightarrow \{A+1,A+2,\cdots,A+B\}$ 
and $\mu_k^{(2)}:\mathcal{W}^B \rightarrow \mathcal{S}$ as follows: if $\mu_k(\underline{w})=(u,\gamma)$, then 
$\mu_k^{(1)}(\underline{w})=u$ and $\mu_k^{(2)}(\underline{w})=\gamma$. 

Our problem is to minimize the long-run average cost per step 
(see equation (5.33) of \cite{bertsekas07dynamic-programming-optimal-control-2} for definition) as follows:

\footnotesize
\begin{equation}
 \inf_{\pi \in \Pi} \, \,  \limsup_{n \rightarrow \infty} \frac{\sum_{k=1}^n \mathbb{E}_{\mu_k} 
c \bigg(\underline{W}(k),\mu_k^{(1)}(\underline{W}(k)),\mu_k^{(2)}(\underline{W}(k))\bigg)}{\sum_{k=1}^n \mathbb{E}_{\pi} \mu_k^{(1)}(\underline{W}(k))} \label{eqn:smdp-problem}
\end{equation}
\normalsize

where $\Pi$ denotes the set of all deterministic, Markov policies, $\pi=\{\mu_i\}_{i \geq 1}$ is a 
specific deterministic, Markov policy and $c(\cdot,\cdot,\cdot)$ is the cost incurred 
when we place a relay (as explained earlier in this section). 
Note that, under any policy, the state evolution process is a positive recurrent Discrete Time Markov 
Chain (DTMC) (under i.i.d shadowing assumption, $\underline{W}(k)$ will be 
i.i.d across $k, k \geq 1$). Also, the state and action spaces are finite. 
Hence, it is sufficient to work with stationary deterministic policies (see \cite{tijms03stochastic-models}).
 
Under our current scenario, the average cost per 
step exists (in fact, the limit exists) and is same for all states, i.e. for all $\underline{w} \in \mathcal{W}^B$.
Let us denote the optimal average cost per step by $\lambda^{*}$.

\subsection{Policy Structure}\label{subsec:smdp-policy-structure}

\begin{thm}\label{thm:policy_structure_smdp_backtracking}
 The optimal action at state $\underline{w}$ in the problem (\ref{eqn:smdp-problem}) is given by: 

\begin{footnotesize}
\begin{eqnarray}
 \mu^*(\underline{w})= \argmin_{u \in \{A+1,\cdots,A+B\}, \gamma \in \mathcal{S}} \bigg( \gamma + \xi_o P_{out}(u,\gamma,w_u) + \xi_r -\lambda^{*}u \bigg) \label{eqn:smdp-optimal-policy}
\end{eqnarray}
\end{footnotesize}
\end{thm}
where $\lambda^*$ is the optimal average cost per step in (\ref{eqn:smdp-problem}).

\begin{proof}
 The optimality equation for the SMDP is given by 
(see \cite{tijms03stochastic-models}, Equation~7.2.2):
\begin{eqnarray}
 v^*(\underline{w}) &=& \min_{u \in \{A+1,\cdots,A+B\}, \gamma \in \mathcal{S}}\bigg\{ \gamma+  \xi_o P_{out}(u,\gamma,w_u)+ \xi_r \nonumber\\
&& -\lambda^{*}u+\sum_{\underline{w}^{'} \in \mathcal{W}^B}g(\underline{w}^{'}) v^*(\underline{w}^{'}) \bigg\}\nonumber\\ \nonumber\\
&& v^*(\underline{w}^{''})=0 \text{ for some } \underline{w}^{''} \in \mathcal{W}^{B}\label{eqn:optimality-smdp}
\end{eqnarray}
where $v^*(\underline{w})$ is the optimal differential cost corresponding to state $\underline{w}$. 
$v^*(\underline{w}^{''})=0 \text{ for some } \underline{w}^{''} \in \mathcal{W}^{B}$ is required to ensure that the 
system of equations in (\ref{eqn:optimality-smdp}) has a unique solution. The structure of the optimal policy is obvious 
from (\ref{eqn:optimality-smdp}), since 
$\sum_{\underline{w}^{'} \in \mathcal{W}^B } g(\underline{w}^{'}) v^*(\underline{w}^{'})$ does not depend on $(u,\gamma)$.
\end{proof}

{\em Remark:} If we take an action $(u,\gamma)$, a cost $( \gamma +  \xi_o P_{out}(u,\gamma,w_u)+ \xi_r)$ will be incurred. 
On the other hand, if 
we incur a cost of $\lambda^*$ over each one of those $u$ steps, the total cost incurred will be $\lambda^*u$. 
The policy selects the placement point that minimizes the difference between these two costs. 
Note that due to the choice of the steps at which measurements are made, the shadowing is 
independent over the steps. This results in each placement point being  a regeneration point in the placement process.

\begin{thm}\label{thm:smdp-cost-vs-xi}
 The optimal average cost $\lambda^{*}$ is jointly concave and increasing in $\xi_r$ and $\xi_o$.
\end{thm}

\begin{proof}
See Appendix~\ref{appendix:backtracking_average_cost}.
\end{proof}

%
%

\subsection{Policy Computation}\label{subsec:smdp-policy-computation}
We adapt a policy iteration from \cite{tijms03stochastic-models} based algorithm 
to calculate $\lambda^{*}$. The algorithm generates a sequence of stationary  policies 
$\{\mu_k\}_{k \geq 1}$ (note that the notation $\mu_k$ was used for a different purpose in 
Section~\ref{subsec:smdp-formulation}; in this subsection each $\mu_k$ is a stationary, deterministic, Markov policy), 
such that for any $k \geq 1$, 
$\mu_k(\underline{w}):\mathcal{W}^B \rightarrow \{A+1,A+2,\cdots,A+B\} \times \mathcal{S}$ maps a state into some action. 
Define the sequence $ \{\mu^{(1)}_k, \mu^{(2)}_k \}_{k \geq 1}$ of functions as 
follows: if $\mu_k(\underline{w})=(u,\gamma)$, then 
$\mu^{(1)}_k (\underline{w})=u$ and $\mu^{(2)}_k (\underline{w})=\gamma$.

\vspace{5mm}
{\bf Policy Iteration based Algorithm:}

{Step~$0$ (Initialization):} Start with an initial stationary deterministic policy $\mu_1$. 

{Step~$1$ (Policy Evaluation):} Calculate the average cost $\lambda_{k}$ corresponding to the policy $\mu_{k}$, for $k \geq 1$. 
This can be done by applying the Renewal Reward Theorem as follows: 

\footnotesize
\begin{equation}
 \lambda_{k}=\frac{  \xi_r+\sum_{\underline{w}} g(\underline{w}) \bigg(\mu^{(2)}_k(\underline{w})+ \xi_o P_{out}(\mu^{(1)}_k(\underline{w}),\mu^{(2)}_k(\underline{w}),w_{\mu^{(1)}_k(\underline{w})})\bigg) } 
{ \sum_{\underline{w}} g(\underline{w}) \mu^{(1)}_k(\underline{w})  ×} \label{eqn:smdp-policy-evaluation}
\end{equation}
\normalsize

{Step~$2$ (Policy Improvement):} Find a new policy $\mu_{k+1}$ by solving the following:

\footnotesize
\begin{eqnarray}
 \mu_{k+1}(\underline{w})=\argmin_{(u,\gamma) } \bigg( \gamma+P_{out}(u,\gamma,w_u)+\xi_r 
-\lambda_{k}u \bigg)\label{eqn:smdp-policy-improvement}
\end{eqnarray}
\normalsize

If $\mu_{k}$ and $\mu_{k+1}$ are the same policy (i.e., if $\lambda_{k-1}=\lambda_k$), then stop and declare $\mu^{*}=\mu_{k}$, $\lambda^{*}=\lambda_{k}$. Otherwise, go to 
Step~$1$.
\qed

{\em Remark:} It was shown in \cite{tijms03stochastic-models} that this policy iteration will converge in a finite number of iterations, for 
finite state and action spaces as in our current problem. The convergence requires that under any stationary policy, the 
state evolves as an irreducible Markov chain, which is satisfied in our current problem.

{\em Computational Complexity:} The state space has cardinality $|\mathcal{W}|^B$, and hence $O(|\mathcal{W}|^B)$ 
addition operations are required to compute $\lambda_{k}$ from (\ref{eqn:smdp-policy-evaluation}). 
However, careful manipulation leads 
to a drastic reduction in this computational requirement, as we will see next.

Note that in (\ref{eqn:smdp-policy-improvement}), if the minimum is achieved by more than one pair of $(u, \gamma)$, 
then any one of them 
can be considered to be the optimal action. Let us use the convention that among all minimizers the pair $(u,\gamma)$ 
with minimum $u$ will be considered 
as the optimal action, and if there are more than one such minimizing pair with same values of $u$, then the pair 
with smallest value of $\gamma$ will be considered. 
We recall that $\mathcal{S}=\{P_1,P_2,\cdots,P_M \}$.  
Let us denote, under policy $\mu_{k+1}$, 
the probability that the optimal control is $(u,\gamma)$ and 
the shadowing is $w$ at the $u$-th location, by $b_k(u,\gamma, w)$. 
Then,

\footnotesize
\begin{eqnarray}
 b_k(u,\gamma,w)&=& \Pi_{r=A+1}^{u-1} \mathbb{P} \bigg(\min_{\gamma^{'} \in \mathcal{S}} (\gamma^{'}+ \xi_o P_{out}(r,\gamma^{'},W_r))-\lambda_k r \nonumber\\
&& >\gamma+\xi_o P_{out}(u,\gamma,w)-\lambda_k u \bigg) \times p_W(w) \nonumber\\
&& \times \Pi_{r=u+1}^{A+B} \mathbb{P} \bigg(\min_{\gamma^{'} \in \mathcal{S}} (\gamma^{'}+ \xi_o P_{out}(r,\gamma^{'},W_r))-\lambda_k r \nonumber\\ 
&& \geq \gamma+\xi_o P_{out}(u,\gamma,w)-\lambda_k u\bigg) \nonumber\\
&& \times \mathbb{I}\bigg\{\gamma=\argmin\{P_1,P_2,\cdots,P_M\}: \nonumber\\
&& \gamma+P_{out}(u,\gamma,w) \nonumber\\
&& =\min_{\gamma^{'}}(\gamma^{'}+\xi_o P_{out}(u,\gamma^{'},w))\bigg\}  \label{eqn:smdp-probability-calculation}
\end{eqnarray}
\normalsize
Now, we can write,

\footnotesize
\begin{eqnarray}
&& \sum_{\underline{w}} g(\underline{w}) \bigg(\mu^{(2)}_k(\underline{w})+\xi_o P_{out}(\mu^{(1)}_k(\underline{w}),\mu^{(2)}_k(\underline{w}),w_{\mu^{(1)}_k(\underline{w})})\bigg) \nonumber\\
&=& \sum_{u=A+1}^{A+B}\sum_{j=1}^{M} \sum_{w \in \mathcal{W}} b_{k-1}(u,P_j,w) \bigg(P_j+\xi_o P_{out}(u,P_j,w)\bigg) \nonumber\\
\label{eqn:smdp-numerator-simpler}
\end{eqnarray}
\normalsize
and 

\footnotesize
\begin{eqnarray}
 \sum_{\underline{w}} g(\underline{w}) \mu^{(1)}_k(\underline{w})
&=& \sum_{u=A+1}^{A+B}\sum_{j=1}^{M} \sum_{w \in \mathcal{W}} b_{k-1}(u,P_j,w) u  \nonumber\\
&=& \sum_{u=A+1}^{A+B} u \sum_{j=1}^{M}\sum_{w \in \mathcal{W}} b_{k-1}(u,P_j,w) \label{eqn:smdp-denominator-simpler}
\end{eqnarray}
\normalsize

Now, for each $(u,\gamma,w)$, $b_{k-1}(u,\gamma,w)$ (in (\ref{eqn:smdp-probability-calculation})) can be computed in 
$O(BM|\mathcal{W}|)$ operations. 
Hence, total number of operations required to compute $b_{k-1}(u,\gamma,w)$ for all $u,\gamma,w$ is $O(B^2 M^2 |\mathcal{W}|^2)$. 
Now, only $O(BM|\mathcal{W}|)$ operations are required in (\ref{eqn:smdp-numerator-simpler}) and (\ref{eqn:smdp-denominator-simpler}). 
Hence, the number of computations required in each iteration is $O(B^2 M^2 |\mathcal{W}|^2)$.

Note that, the policy improvement step is not explicitly required in the policy iteration. This is because in the 
policy evaluation step, $\lambda_k$ is sufficient to compute $b_k(u, \gamma, w)$ for all $u,\gamma,w$ 
and thereby to compute $\lambda_{k+1}$. 
Hence, we need not store the policy in each iteration.

\subsection{No Backtracking}\label{subsec:average-cost-no-backtracking}

When there is no backtracking (i.e., the deployment agent decides at each 
step whether to place a relay or not), the state and action spaces are same as discussed in 
Section~\ref{subsec:problem_formulation_sum_power_sum_outage_no_backtracking}. 
In this section, we are interested in the minimum average cost per step problem, assuming 
that the line has infinite length. The single-stage cost is the same as in 
Section~\ref{sec:sum_power_sum_outage_no_backtracking_discounted}.

Note that the problem 
(\ref{eqn:sum_power_discounted_no_backtracking}) can be considered as an infinite 
horizon discounted cost problem with discount factor 
$(1-\theta)$. Hence, keeping in mind that we have finite state and action spaces, we observe that  
for the discount factor sufficiently close to $1$, i.e., for $\theta$ sufficiently close to $0$, 
the optimal policy for problem (\ref{eqn:sum_power_discounted_no_backtracking}) is optimal for the 
problem (\ref{eqn:smdp-problem}) 
(see \cite{bertsekas07dynamic-programming-optimal-control-2}, Proposition~4.1.7). In particular, the optimal average 
cost per step with no backtracking, $\lambda'$, is given by $\lambda'=\lim_{\theta \rightarrow 0}\theta J_{\theta}(\mathbf{0})$ 
 (see \cite{bertsekas07dynamic-programming-optimal-control-2}, Section~4.1.1), 
where $J_{\theta}(\mathbf{0})$ is the optimal cost for problem (\ref{eqn:sum_power_discounted_no_backtracking}) with backtracking 
with the probability of the line ending in the next step is $\theta$.

\begin{thm}\label{theorem:comparison-backtracking-no-backtracking}
 $\lambda' \geq \lambda^*$.
\end{thm}

\begin{proof}
 See Appendix~\ref{appendix:backtracking_average_cost}.
\end{proof}

\section{Numerical Work}\label{sec:numerical_work}

\subsection{Parameter Values}\label{subsec:parameter_values}

Recall the notation used in Section~\ref{sec:system_model_and_notation}. 
We consider deployment along a line with step size $\delta=6$~meters, $A=5$, $B=5$ and $\theta=0.04$ 
(mean length of the line is $25$~steps, i.e., $150$~meters). 
The set of transmit power levels $\mathcal{S}$ is taken to be $\{-25,-15,-10,-5,0\}$~dBm. 
For the channel model as in 
(\ref{eqn:channel_model}), we consider path-loss exponent $\eta=3.8$ and $c=10^{0.00054}$. Fading is assumed to be Rayleigh; 
$H \sim Exponential(1)$. Shadowing $W$ is assumed to be log-normal with $W=10^{\frac{Y}{10×}}$ with 
$Y \sim \mathcal{N}(0, \sigma^2)$ where $\sigma=7$~dB. The values of the parameters in the channel model 
were estimated from data obtained by experiments (using $2.2$~dBi antennas in the transmitter and the receiver) 
in a forest-like environment inside our campus. However, for the purpose 
of numerical computation we assume that $Y$ can take values in the interval $[-4 \sigma, 4 \sigma]$ in steps of $0.02$. Thus we 
have converted the probability density function of $Y$ into the probability mass function of a discrete-valued random variable, and 
the probability of $Y$ being outside the interval $[-4 \sigma, 4 \sigma]$ is negligible ($6.3342 \times 10^{-5}$). This 
discretization renders the state space finite for each problem. We define outage to be the event when the received signal 
power of a packet falls below $P_{rcv-min}=10^{-8.8}$~mW ($-88$~dBm). For a commercial implementation of the PHY/MAC of
IEEE~$802.15.4$ (a popular wireless sensor networking standard), $-88$~dBm received power corresponds to a $2\%$ packet
loss probability for $140$~byte packets.

\vspace{-2mm}
\subsection{Geometrically distributed distance $L$ to the source; no backtracking}\label{subsec:numerical_no_backtracking_discounted}
\vspace{-1mm}
\subsubsection{Sum-Power, Sum-Outage Objective; Policy Structure}\label{subsubsec:policy_structure_numerical_sum_power_sum_outage_no_backtracking_discounted}

\begin{figure}[!t]
\centering
\includegraphics[scale=0.24]{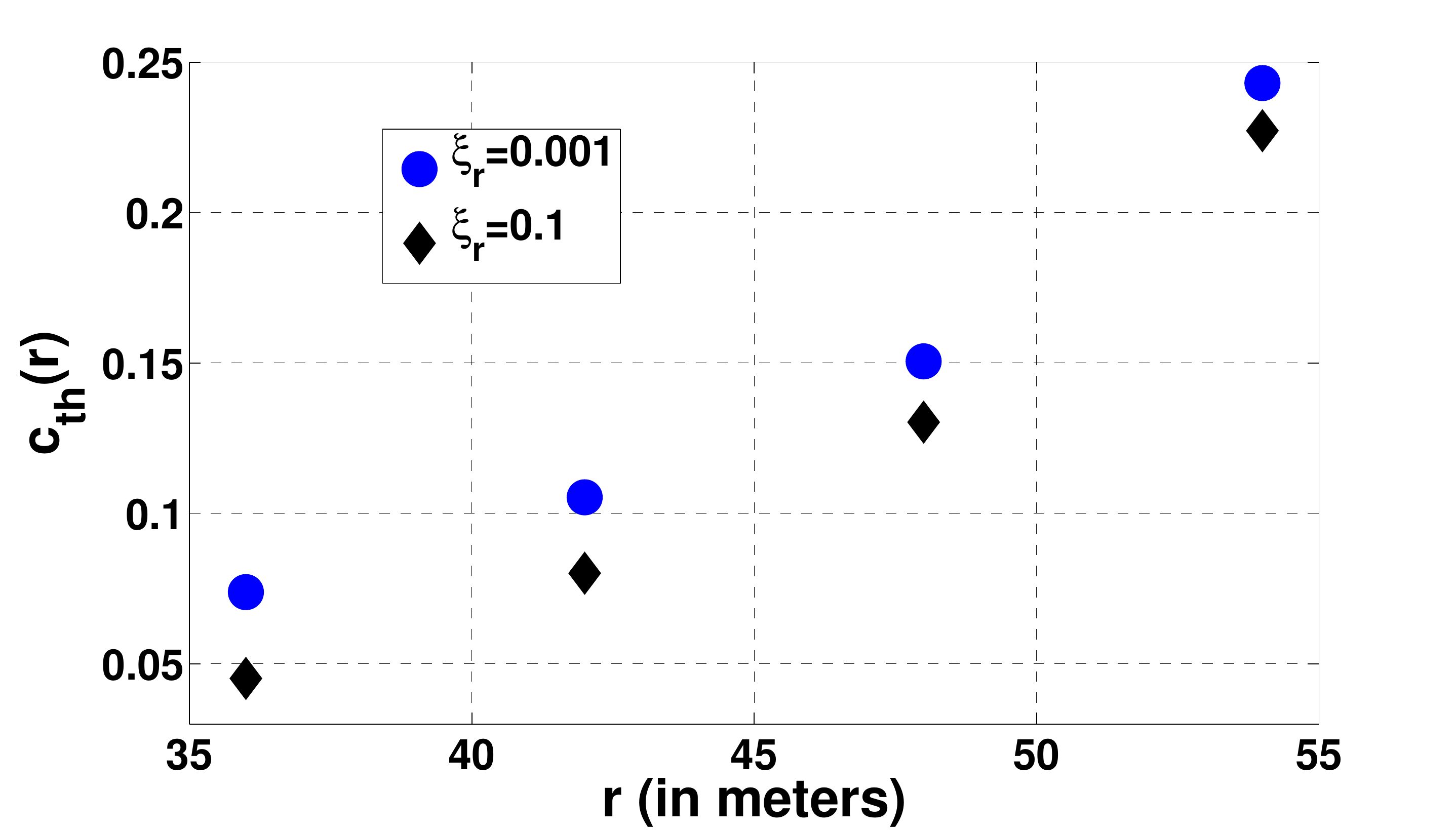}
\vspace{-1mm}
\caption{As-you-go deployment without backtracking; variation of $c_{th}(r)$ with $r$ for $\xi_o=1$ and various values of $\xi_r$.}
\label{fig:threshold_vs_distance_various_relay_cost}
\vspace{-4mm}
\end{figure}

\begin{figure}[!t]
\centering
\includegraphics[scale=0.24]{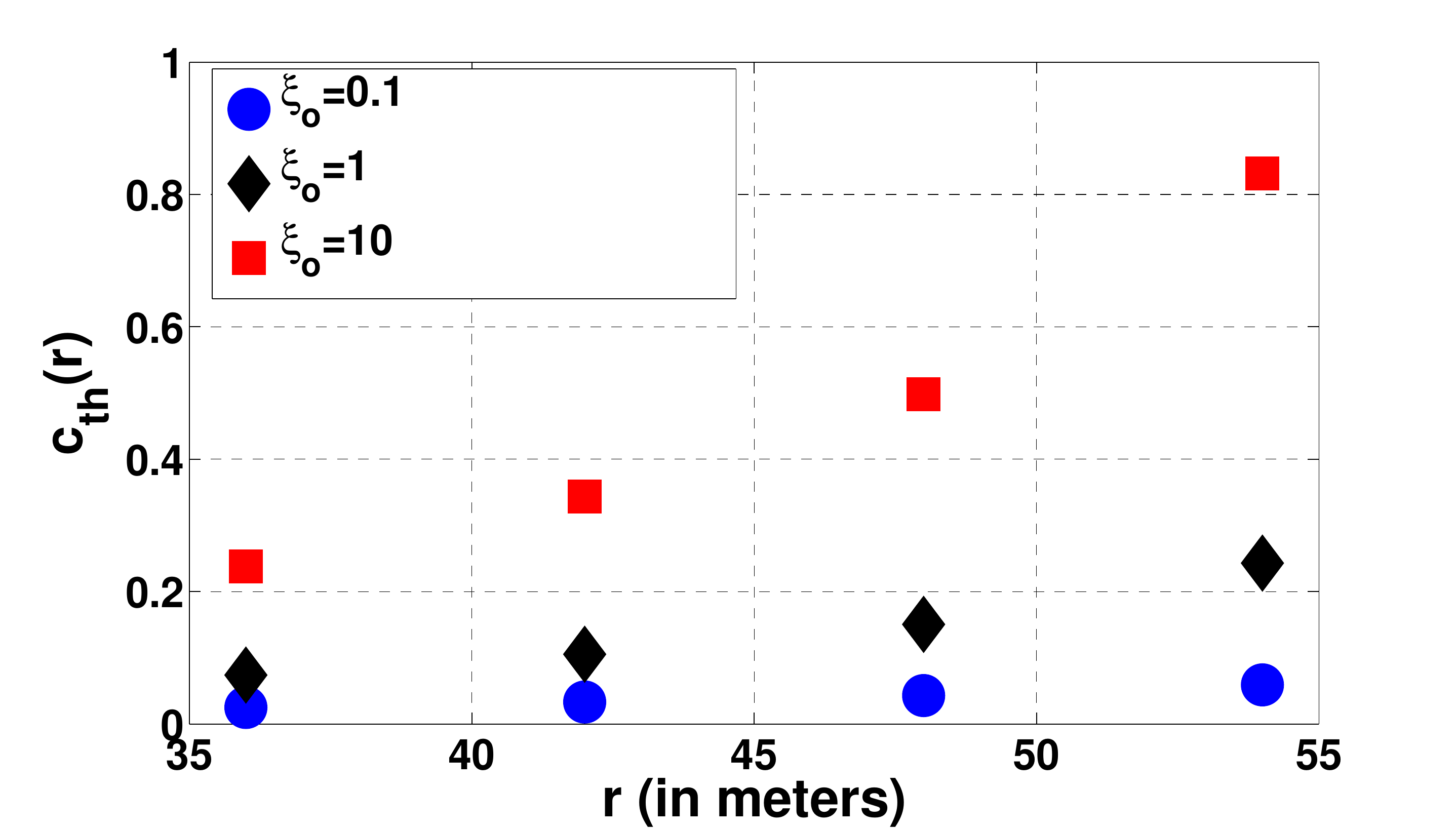}
\vspace{-2mm}
\caption{As-you-go deployment without backtracking; variation of $c_{th}(r)$ with $r$ for $\xi_r=0.001$ and various values of $\xi_o$.}
\label{fig:threshold_vs_distance_various_outage_cost}
\vspace{-5mm}
\end{figure}
    
The variation of $c_{th}(r)$ (for Problem~(\ref{eqn:sum_power_discounted_no_backtracking})) 
with the relay cost $\xi_r$ and the cost of outage $\xi_o$ has been shown 
in Figure~\ref{fig:threshold_vs_distance_various_relay_cost} and Figure~\ref{fig:threshold_vs_distance_various_outage_cost}. 
For a fixed $\xi_o$, $c_{th}(r)$ decreases with $\xi_r$; i.e., as the cost of placing a relay increases, we place relays less 
frequently. On the other hand, for a fixed $\xi_r$, $c_{th}(r)$ increases with $\xi_o$. This happens because if the cost of 
outage increases, we cannot tolerate outage and place the relays close to each other.

\subsubsection{Comparison between the total costs of the sum power and the max power problem}
\label{subsubsec:numerical_comparison_sum_power_max_power_no_backtracking_discounted}

\begin{table}[t!]
\centering
\footnotesize
\begin{tabular}{|c |c |c |c|}
\hline
 $\xi_r$ & $\xi_o$    &  Optimal cost    &   Optimal cost    \\ 
          &            &   for Sum Power    & for Max Power \\ \hline
0.001   &  0.1  & 0.0926  &   0.0472  \\ \hline
0.001   & 1   &  0.2646 &  0.1442  \\ \hline
0.001   & 10   & 0.8177  &  0.4532  \\ \hline
0.01   &  0.1  & 0.1182  & 0.0757   \\ \hline
0.01   &  1  & 0.2925  & 0.1734   \\ \hline
0.01   & 10   & 0.8457  &  0.4826  \\ \hline
\end{tabular}
\normalsize
\caption{Geometrically distributed distance to the source with $\theta=0.04$: comparison of the optimal cost 
without backtracking between problems (\ref{eqn:sum_power_discounted_no_backtracking}) and 
(\ref{eqn:max_power_discounted_no_backtracking}), for 
the parameters in Section~\ref{subsec:parameter_values}, for various values of $\xi_r$ and $\xi_o$.}
\vspace{-0.8cm}
\label{table:comparison_discounted_sum_power_vs_max_power_no_backtracking}
\end{table}

Table~\ref{table:comparison_discounted_sum_power_vs_max_power_no_backtracking} compares the optimal total costs 
without backtracking of the 
problems (\ref{eqn:sum_power_discounted_no_backtracking}) and (\ref{eqn:max_power_discounted_no_backtracking}), for 
various values of $\xi_o$ and $\xi_r$. The first problem always has higher cost  
(Theorem~\ref{thm:comparison_costs_max_power_vs_sum_power}), 
since the sum power in a network is always greater than the max power.

\vspace{-2mm}
\subsection{Geometrically distributed distance $L$ to source; with and without backtracking}\label{subsec:comparison_cost_backtracking_no_backtracking_discounted}
\vspace{-2mm}

\begin{table}[t!]
\centering
\footnotesize
\begin{tabular}{|c |c |c |c|}
\hline
 $\xi_r$ & $\xi_o$    &  Optimal cost    &   Optimal cost    \\ 
          &            &   without backtracking    & with backtracking \\ \hline
0.001   & 0.1   &  0.0926  &  0.0581  \\ \hline
0.001   & 1   &  0.2646  & 0.1502   \\ \hline
0.001   & 10 & 0.8177   &  0.4650  \\ \hline
0.01   &  0.1  & 0.1182   & 0.0806   \\ \hline
0.01   &  1  & 0.2925   &  0.1728  \\ \hline
0.01   &  10  &  0.8457  &  0.4878  \\ \hline
\end{tabular}
\normalsize
\caption{Sum power objective; geometrically distributed distance to the source; with and without backtracking; 
comparison of the optimal cost for various values of $\xi_r$ and $\xi_o$.}
\vspace{-0.8cm}
\label{table:comparison_discounted_sum_power_sum_outage_backtracking_no_backtracking}
\end{table}

The comparison between the optimal cost of as-you-go deployment with and without backtracking, for 
Problem~(\ref{eqn:sum_power_discounted_no_backtracking}), for various values of $\xi_r$ and $\xi_o$, 
and for parameter values as in Section~\ref{subsec:parameter_values}, 
are shown in Table~\ref{table:comparison_discounted_sum_power_sum_outage_backtracking_no_backtracking}. It is obvious that 
backtracking can provide significant reduction in the cost compared to no backtracking, 
due to the fact that in backtracking we choose the best relay location among many (similar arguments as in 
Theorem~\ref{theorem:comparison-backtracking-no-backtracking} works here).

\begin{table}[t!]
\centering
\footnotesize
\begin{tabular}{|c |c |c |c|c|}
\hline
 $\xi_r$ & $\xi_o$    &  mean power     &   mean hop   & Mean outage   \\ 
         &           &  per hop   &     length   &    probability   \\ 
          &           &   (in mW)  &    (in steps)  &   per link  \\  \hline     
0.001   &  0.1   &    0.0092            &   7.5965                 &   0.1157            \\ \hline
0.001   &  1   &   0.0311             &    7.6260                &    0.0251           \\ \hline
0.001   &  10        &   0.0842             &  7.5445                  &  0.0085             \\ \hline
0.01   &  0.1   &  0.0097              &  7.7576                  &  0.1160             \\ \hline
0.01   &  1        &  0.0312              &   7.6900                 &  0.0254             \\ \hline
0.01   &  10      &  0.0844              &     7.5645              &    0.0085           \\ \hline
0.1   &  0.01     &  0.0032              &    10.0000                &    0.7856           \\ \hline
0.1   & 0.1    &  0.0191              &    9.0787                &    0.1382            \\ \hline
0.1   & 1   &  0.0332              &     8.1944               &  0.0305             \\ \hline
0.1   &  10   &  0.0869              &   7.7556                 &     0.0089          \\ \hline
\end{tabular}
\normalsize
\caption{Average cost per step objective with backtracking: mean power per link, mean outage probability per 
link and the mean hop length under the optimal 
policy; various values of $\xi_r$, $\xi_o$.}
\vspace{-0.8cm}
\label{table:various_cost_components_average_cost_backtracking}
\end{table}

\vspace{-2mm}
\subsection{Average cost per step; sum power and sum outage; with and without backtracking}\label{subsec:comparison_average_cost_backtracking_no_backtracking_heuristic}
\vspace{-2mm}

\begin{table}[t!]
\centering
\footnotesize
\begin{tabular}{|c |c |c |c|c|}
\hline
 $\xi_r$ & $\xi_o$    &  $\lambda^{*}$    &   $\lambda^{'}$  & $\lambda_{h}$  \\ \hline
0.001   &  0.1   &  0.0029  & 0.0035  & 0.0029  \\ \hline
0.001   &  1  &  0.0075  & 0.0100   &  0.0075 \\ \hline
0.001   &  10  & 0.0226   & 0.0307 &  0.0228 \\ \hline
0.01   &  0.1   &  0.0040  & 0.0047  & 0.0041  \\ \hline
0.01   &  1  &  0.0087  & 0.0113  &  0.0087 \\ \hline
0.01   &  10  &  0.0238  & 0.0321  &  0.0239 \\ \hline
0.1   &  0.01  & 0.0111   & 0.0111  & 0.0111  \\ \hline
0.1   & 0.1    &  0.0146  & 0.0155  & 0.0147  \\ \hline
0.1   & 1   &  0.0200  & 0.0238  &  0.0200 \\ \hline
0.1   &  10  &  0.0355  & 0.0450  &  0.0357 \\ \hline
\end{tabular}
\normalsize
\caption{Average cost per step objective: as-you-go deployment with and without backtracking and for a heuristic; 
various values of $\xi_r$ and $\xi_o$.}
\vspace{-1.3cm}
\label{table:comparison_average_cost_backtracking_no_backtracking_heuristic}
\end{table}

$\lambda^{*}$ in 
Table~\ref{table:comparison_average_cost_backtracking_no_backtracking_heuristic} denotes the optimal average cost per step 
with backtracking, as discussed in Theorem~\ref{thm:policy_structure_smdp_backtracking}. $\lambda'$ denotes the 
optimal average cost per step without backtracking, as discussed in Section~\ref{subsec:average-cost-no-backtracking}. 
$\lambda_{h}$ is the optimal average cost per step for the following heuristic policy. Recall the notation used in 
Section~\ref{sec:backtracking_average_cost}. The heuristic policy solves the problem (at state $\underline{w}$) 
$\min_{u \in \{A+1,\cdots,A+B\}, \gamma \in \mathcal{S}} \frac{\gamma+\xi_o P_{out}(u,\gamma,w_u)+\xi_r}{u×}$ to 
select the placement location and the transmit power level to use. Note that this heuristic, unlike our earlier 
policies, does not require any channel model to make the placement decision (e.g., we need not know explicitly the 
values of $\eta$, $\sigma$ etc., as we had required earlier to compute $\lambda^{*}$). In this heuristic policy, 
the deployment agent, at each $u \in \{A+1,\cdots,A+B\}$, measures for each $\gamma \in \mathcal{S}$ the outage 
probability to the previous node (without using the model to calculate shadowing). Then he performs 
$\min_{u, \gamma} \frac{\gamma+\xi_o P_{out}(u,\gamma,w_u)+\xi_r}{u×}$ to make the 
placement decision. Thus, the heuristic policy focuses on minimizing the per-step cost 
over the new link. 

Table~\ref{table:various_cost_components_average_cost_backtracking} shows the mean power per link, 
the mean distance between two consecutive nodes, and the mean outage probability per link under the optimal policy 
with backtracking. Note that for some cases (e.g., $\xi_r=0.1$, $\xi_o=0.01$), the relay is always placed at the 
$10$-th step (step~$(A+B)$) and uses $0.0032$~mW (i.e., $-25$~dBm) power, but this renders the outage probability very high. 
However, for each of $\xi_r=0.001,0.01,0.1$, we have reasonably small outage probability for 
higher values of the outage cost ($\xi_o=1,10$). For each $\xi_r$, as $\xi_o$ increases, the outage probability decreases, the 
mean power per link increases (to reduce the outage probability) and the relays are placed closer and closer to each other.

From Table~\ref{table:comparison_average_cost_backtracking_no_backtracking_heuristic}, we find that $\lambda^{*}$ is in general 
substantially smaller than $\lambda^{'}$, except for some special cases where we always place at (or near) 
the $10$-th step and use $-25$~dBm transmit power (the optimal policy without backtracking also does the same in such cases). 
All that it says that by backtracking we can save substantial amount of cost, 
though it will require some additional walking and measurements. However, we notice that 
$\lambda_{h}$ is always equal to or very close to $\lambda^{*}$. 
{\em This shows that this model-free heuristic policy can perform as a very 
good suboptimal policy.}

\section{Conclusion}\label{sec:conclusion}
In this paper, we have developed several approaches for as-you-go deployment of wireless relay
networks assuming very light traffic, using on-line measurements, and permitting backtracking. Each problem was 
formulated as an MDP and its optimal policy structure was studied. 
Numerical results have been provided to illustrate the performance and tradeoffs, and a nice heuristic policy 
was proposed for the average cost per step problem with backtracking. This 
work can be extended in several ways: (i) We could design a more robust network 
by asking for each relay to have multiple neighbours, 
(ii) It may be noted that even though our design approach assumes the lone packet traffic model, 
the network thus obtained will be able to carry a certain amount of positive traffic. Can the design process be modified to increase 
network capacity? All these aspects are problems that we are currently pursuing.

\renewcommand{\thesubsection}{\Alph{subsection}}

\appendices

\section{Impromptu Deployment for Geometrically Distributed Length without Backtracking: Sum Power and 
Sum Outage Objective}\label{appendix:sum_power_sum_outage_no_backtracking_discounted}

\textbf{Proof of Lemma~\ref{lemma:value_iteration_sum_power_sum_outage_no_backtracking}}
Here we have an infinite horizon total cost MDP with finite state space and finite action space. The assumption P of 
Chapter $3$ in \cite{bertsekas07dynamic-programming-optimal-control-2} is satisfied since the single-stage cost is nonnegative. 
Hence, by combining Proposition $3.1.5$ and Proposition $3.1.6$ of \cite{bertsekas07dynamic-programming-optimal-control-2}, 
we obtain the result.

\textbf{Proof of Lemma~\ref{lemma:value_function_properties_sum_power_sum_outage_no_backtracking}}
Note that the function $J^{(0)}(\cdot):=0$ satisfies all the assertions. Let us assume, as our induction 
hypothesis, that $J^{(k)}(\cdot)$ satisfies 
all the assertions. Now $P_{out}(r,\gamma,w)$ is increasing in $r$ and decreasing in $w$ (by our 
channel modeling assumptions in Section~\ref{subsection:channel_model}), 
and the single stage 
costs are linear (hence concave) increasing in $\xi_r$, $\xi_o$. 
Then from the value iteration (\ref{eqn:value_iteration_sum_power_sum_outage_no_backtracking}), 
$J^{(k+1)}(r,w)$ is pointwise minimum of functions which are increasing in $r$, $\xi_o$ and $\xi_r$, 
decreasing in $w$, and jointly concave in $\xi_o$ and $\xi_r$. Similarly, $J^{(k+1)}(\mathbf{0})$ is also pointwise minimum 
of functions which are increasing and jointly concave in $\xi_r$ and $\xi_o$. 
Hence, the assertions hold for $J^{(k+1)}(\mathbf{0})$. 
Since $J^{(k)}(\cdot) \uparrow J(\cdot)$, the results follow.

\textbf{Proof of Theorem~\ref{theorem:policy_structure_sum_power_sum_outage_no_backtracking}}
Consider the Bellman equation (\ref{eqn:bellman_equation_sum_power_sum_outage_no_backtracking}). We will place a relay at state 
$(r,w)$ iff the cost of placing a relay, i.e., 
$\min_{\gamma \in \mathcal{S}} (\gamma+\xi_o P_{out}(r,\gamma,w) )+\xi_r + J(\mathbf{0})$ is less than or equal to 
the cost of not placing, i.e., $\theta \mathbb{E}_W \min_{\gamma \in \mathcal{S}} (\gamma+ \xi_o  P_{out} (r+1,\gamma,W)) + 
(1-\theta)\mathbb{E}_W J(r+1,W)$. Hence, it is obvious 
that we will place a relay at state $(r,w)$ iff 
$\min_{\gamma \in \mathcal{S}} (\gamma+\xi_o P_{out}(r,\gamma,w) ) \leq c_{th}(r)$ where the threshold $c_{th}(r)$ is given by:
\begin{eqnarray}
 c_{th}(r)&=&\theta \mathbb{E}_W \min_{\gamma \in \mathcal{S}} (\gamma+ \xi_o  P_{out} (r+1,\gamma,W)) \nonumber\\
&& + (1-\theta)\mathbb{E}_W J(r+1,W) -(\xi_r + J(\mathbf{0})) \label{eqn:c_th_r_expression}
\end{eqnarray}
By Proposition~$3.1.3$ of \cite{bertsekas07dynamic-programming-optimal-control-2}, if there
exists a stationary policy $\{\mu,\mu,\cdots\}$ such that for each state, the action chosen by the policy is the action that
achieves the minimum in the Bellman equation, then that stationary policy will be an optimal policy, i.e., 
the minimizer in Bellman equation gives the optimal action. Hence, if the decision is to place a relay at state 
$(r,w)$, then the power has to be chosen as 
$\argmin_{\gamma \in \mathcal{S}} \bigg(\gamma+\xi_o P_{out}(r,\gamma,w)\bigg)$. 

Since $P_{out}(r,\gamma,w)$ and $J(r,w)$ is increasing in $r$ for each $\gamma,w$, it is easy to see 
that $c_{th}(r)$ is increasing in $r$.

\section{Impromptu Deployment for Geometrically Distributed Length without Backtracking: Max Power and 
Sum Outage Objective}\label{appendix:max_power_sum_outage_no_backtracking_discounted}

\textbf{Proof of Lemma~\ref{lemma:value_function_properties_max_power_sum_outage_no_backtracking}}
Note that the function $J^{(0)}(\cdot):=0$ satisfies all the assertions. Let us assume, as our induction 
hypothesis, that $J^{(k)}(\cdot)$ satisfies 
all the assertions. Now $P_{out}(r,\gamma,w)$ is increasing in $r$ and decreasing in $w$ (by our 
channel modeling assumptions in Section~\ref{subsection:channel_model}), 
and the single stage 
costs are linear (hence concave) increasing in $\xi_r$, $\xi_o$ and also increasing in $\gamma_{max}$. 
Then from the value iteration (\ref{eqn:value_iteration_max_power_sum_outage_no_backtracking}), 
$J^{(k+1)}(r,w,\gamma_{max})$ is pointwise minimum of functions which are increasing in $r$, $\gamma_{max}$, $\xi_o$ and $\xi_r$, 
decreasing in $w$, and jointly concave in $\xi_o$ and $\xi_r$.  
$J^{(k+1)}(\mathbf{0};\gamma_{max})$ is the sum of pointwise minimum 
of functions which are increasing and jointly concave in $\xi_r$ and $\xi_o$. $J^{(k+1)}(\mathbf{0};\gamma_{max})$ 
is the sum of increasing functions of $\gamma_{max}$. 
Hence, the assertions hold for $J^{(k+1)}(\mathbf{0};\gamma_{max})$. 
Since $J^{(k)}(\cdot) \uparrow J(\cdot)$, the results follow.

\textbf{Proof of Theorem~\ref{theorem:policy_structure_max_power_sum_outage_no_backtracking}}
Consider the Bellman equation (\ref{eqn:bellman_equation_max_power_sum_outage_no_backtracking}). We will place a relay at state 
$(r,w,\gamma_{max})$ iff the cost of placing a relay  
$c_p:=\min_{\gamma \in \mathcal{S}} \bigg( \xi_o P_{out}(r,\gamma,w)+\xi_r 
+ J(\mathbf{0};\max\{\gamma,\gamma_{max}\}) \bigg)$, is less than or equal to 
the cost of not placing  
$c_{np}:=\theta \mathbb{E}_W \min_{\gamma \in \mathcal{S}} \bigg(\max\{\gamma,\gamma_{max}\} + \xi_o  P_{out} (r+1,\gamma,W)\bigg) 
 + (1-\theta)\mathbb{E}_W J(r+1,W,\gamma_{max}) $. This yields the condition that 
$\argmin_{\gamma \in \mathcal{S}} \bigg( \xi_o P_{out}(r,\gamma,w)+ J(\mathbf{0};\max\{\gamma,\gamma_{max}\}) \bigg) \leq c_{np}-\xi_r:=c_{th}(r,\gamma_{max})$. 
Since $c_{np}$ is increasing in $r$ and $\gamma_{max}$, $c_{th}(r,\gamma_{max})$ increases in $r,\gamma_{max}$. 
Also, the minimizer in Bellman equation gives the optimal action 
(by the same arguments as in the proof of Theorem~\ref{theorem:policy_structure_sum_power_sum_outage_no_backtracking}). 
Hence, if the decision is to place a relay at state 
$(r,w,\gamma_{max})$, then the power has to be chosen as 
$\argmin_{\gamma \in \mathcal{S}} \bigg( \xi_o P_{out}(r,\gamma,w)+\xi_r + J(\mathbf{0};\max\{\gamma,\gamma_{max}\}) \bigg)$.

\section{Impromptu Deployment for Geometrically Distributed Length with Backtracking: Sum Power and 
Sum Outage Objective}\label{appendix:sum_power_sum_outage_with_backtracking_discounted}

\begin{figure}[!t]
\centering
\includegraphics[scale=0.4]{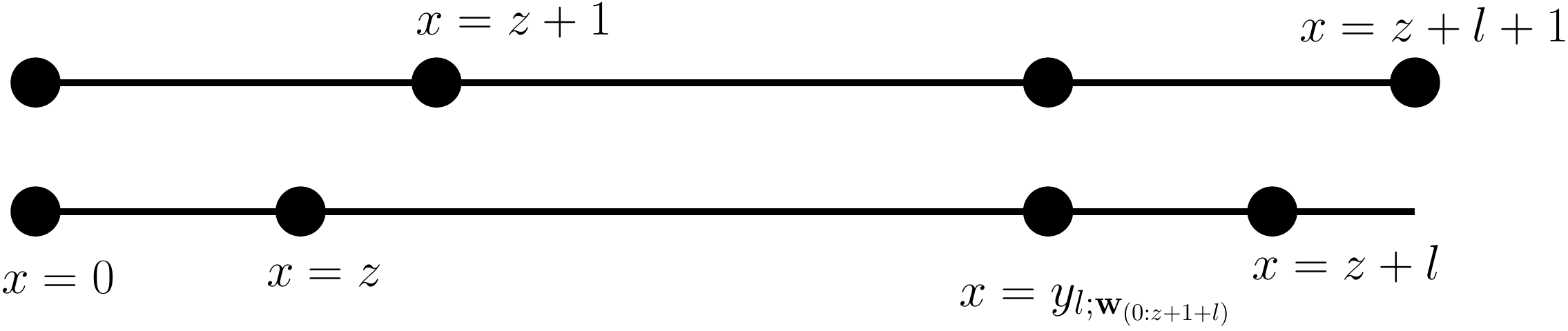}
\vspace{-2mm}
\caption{A diagram illustrating the idea behind the proof of 
Lemma~\ref{lemma:properties_J_vs_z_sum_power_sum_outage_with_backtracking}.}
\label{fig:proof_illustration}
\vspace{-5mm}
\end{figure}

\textbf{Proof of Lemma~\ref{lemma:properties_J_vs_z_sum_power_sum_outage_with_backtracking}}
We will show that $J(z+1;\mathbf{0}) \geq J(z;\mathbf{0})$. Consider two instances of the deployment process where the agent 
has placed a relay 
at his current location. The deployment over the remaining part of the line depends on two things: 
(i) the shadowing realizations in all the links 
over the rest of the line, (ii) the residual length of the line. Suppose that the deployment is being 
done along $x$-axis and that the current location of the deployment agent is $x=0$ in both instances. For the first instance 
the residual length of the line is $(z+1)+L_1$ where $L_1 \sim Geometric(\theta)$ and for the second instance 
the residual length is $(z+L_2)$ where $L_2 \sim Geometric(\theta)$. 

Let us denote the optimal policy for the first instance by $\mu_{z+1}^{*}$ and that for the 
second instance by $\mu_{z}^{*}$. We will prove that 
$J(z+1;\mathbf{0}):=J_{\mu_{z+1}^{*}}(z+1;\mathbf{0}) \geq J_{\mu_{z+1}^{*}}(z;\mathbf{0}) \geq J_{\mu_{z}^{*}}(z;\mathbf{0})=:J(z;\mathbf{0})$, 
where $J_{\mu}(s)$ is the cost-to-go under policy $\mu$ from the state $s$. 

Note that any pair of potential relay locations of the form $\{(i,j): i>j, i \in \{0,1,2,\cdots\}, j \in \{0,1,2,\cdots\}, 
|i-j| \leq (A+B)\}$ is a possible link. Let us denote the shadowing component of the path-loss over 
in link $(i,j)$ by $w_{i,j}$. Let us denote 
the collection of the random shadowing of all possible links emanating from and ending at 
$\{j,j+1,\cdots,i\}$ ($i>j$) by $\mathbf{W}_{(j:i)}$, and that 
of all possible links emanating from the segment $\{j,j+1,\cdots,i\}$ and ending at $\{j_1,j_1+1,\cdots,i_1\}$ 
by $\mathbf{W}_{(j:i);(j_1:i_1)}$. Let us also denote the realizations of these collection of random variables by 
$\mathbf{w}_{(j:i)}$ and $\mathbf{w}_{(j:i);(j_1:i_1)}$ respectively. Now, for $L_1=l$ and the realization of the shadowing 
$\mathbf{w}_{(0:z+1+l)}$ of all possible links, let us denote the location of the last placed relay in the 
first instance (excluding the source at $x=z+1+l$) under policy $\mu_{z+1}^{*}$ by $x=y_{l;\mathbf{w}_{(0:z+1+l)}}$, 
and the cost incurred over the links solely in the 
locations $\{j,j+1,\cdots,i\}$ by $c(l;\mathbf{w}_{(0:z+1+l)};j:i)$.

The idea behind the proof is as follows. Consider Figure~\ref{fig:proof_illustration} which depicts the two instances of the 
problem. Consider the case where we fix $L_1=l$, $L_2=l$ and the shadowing of all possible links between $x=0$ and $x=z+1+l$ 
are also fixed and they are the same for both instances. Note that the location of the last placed relay 
(before the source) for the first instance under the policy $\mu_{z+1}^{*}$ is 
denoted by $y_{l;\mathbf{w}_{(0:z+1+l)}}$. If $y_{l;\mathbf{w}_{(0:z+1+l)}}=z+l$, then the source 
placed at $x=z+l$ in the second instance will use 
the same transmit power as used by the relay placed at $x=z+l$ in the first instance; in fact, in the 
region between $x=0$ and $x=z+l$ we will have the same placement locations, power and outage costs 
in both instances. But then there will be an extra link in the first instance 
from $x=z+l+1$ to $x=z+l$, and hence the first instance will have more cost. On the other hand, if $y_{l;\mathbf{w}_{(0:z+1+l)}}<z+l$, then 
in the region between $x=0$ and $x=y_{l;\mathbf{w}_{(0:z+1+l)}}$ we will have the same power and outage cost and same placement 
locations in both instances. But the last link in the first instance 
has length $(z+1+l-y_{l;\mathbf{w}_{(0:z+1+l)}})$, and that in the second instance has length $(z+l-y_{l;\mathbf{w}_{(0:z+1+l)}})$. If we now take expectation of the costs of these 
two links in two different cases over the shadowing in all possible links between $x=y_{l;\mathbf{w}_{(0:z+1+l)}}$ and $x=z+1+l$, 
then the link in the first instance 
will have higher expected cost since it is longer. This will happen for every possible values of $l$ and $y_{l;\mathbf{w}_{(0:z+1+l)}}$.

Now we will formally prove this lemma. By total probability theorem, we can write,

\footnotesize
\begin{eqnarray*}
&& J_{\mu_{z+1}^{*}}(z+1;\mathbf{0})\nonumber\\
&=& \sum_{l=1}^{\infty}(1-\theta)^{l-1}\theta \sum_{\mathbf{w}_{(0:z+1+l)}} p_{\mathbf{W}_{(0:z+1+l)}}(\mathbf{w}_{(0:z+1+l)}) \nonumber\\
&& \times \bigg(c(l;\mathbf{w}_{(0:z+1+l)};0:y_{l;\mathbf{w}_{(0:z+1+l)}}) + \nonumber\\
&& c(l;\mathbf{w}_{(0:z+1+l)};y_{l;\mathbf{w}_{(0:z+1+l)}}:z+1+l )\bigg) \nonumber\\
&=& \sum_{l=1}^{\infty}(1-\theta)^{l-1}\theta \sum_{\mathbf{w}_{(0:z+1+l)}} p_{\mathbf{W}_{(0:z+1+l)}}(\mathbf{w}_{(0:z+1+l)}) \nonumber\\
&& \times  \sum_{k=1}^{z+l} \mathbb{I}(y_{l;\mathbf{w}_{(0:z+1+l)}}=k) \bigg(c(l;\mathbf{w}_{(0:z+1+l)};0:k) + \nonumber\\
&& c(l;\mathbf{w}_{(0:z+1+l)};k:z+1+l )\bigg) \nonumber\\
&=& \sum_{l=1}^{\infty}(1-\theta)^{l-1}\theta \sum_{\mathbf{w}_{(0:z+1+l)}} p_{\mathbf{W}_{(0:z+1+l)}}(\mathbf{w}_{(0:z+1+l)}) \nonumber\\
&& \times  \sum_{k=1}^{z+l-1} \mathbb{I}(y_{l;\mathbf{w}_{(0:z+1+l)}}=k) \bigg(c(l;\mathbf{w}_{(0:z+1+l)};0:k) + \nonumber\\
\end{eqnarray*}
\begin{eqnarray}
&& c(l;\mathbf{w}_{(0:z+1+l)};k:z+1+l )\bigg) \nonumber\\
&&+ \sum_{l=1}^{\infty}(1-\theta)^{l-1}\theta \sum_{\mathbf{w}_{(0:z+1+l)}} p_{\mathbf{W}_{(0:z+1+l)}}(\mathbf{w}_{(0:z+1+l)}) \nonumber\\
&& \times   \mathbb{I}(y_{l;\mathbf{w}_{(0:z+1+l)}}=z+l) \bigg(c(l;\mathbf{w}_{(0:z+1+l)};0:z+l) + \nonumber\\
&& c(l;\mathbf{w}_{(0:z+1+l)};z+l:z+1+l )\bigg) 
\end{eqnarray}
\normalsize
On the other hand, 

\footnotesize
\begin{eqnarray}
 && J_{\mu_{z+1}^{*}}(z;\mathbf{0})\nonumber\\
&=& \sum_{l=1}^{\infty}(1-\theta)^{l-1}\theta \sum_{\mathbf{w}_{(0:z+1+l)}} p_{\mathbf{W}_{(0:z+1+l)}}(\mathbf{w}_{(0:z+1+l)}) \nonumber\\
&& \times  \sum_{k=1}^{z+l-1} \mathbb{I}(y_{l;\mathbf{w}_{(0:z+1+l)}}=k) \bigg(c(l;\mathbf{w}_{(0:z+1+l)};0:k) + \nonumber\\
&& c(l;\mathbf{w}_{(0:z+1+l)};k:z+l )\bigg) \nonumber\\
&&+ \sum_{l=1}^{\infty}(1-\theta)^{l-1}\theta \sum_{\mathbf{w}_{(0:z+1+l)}} p_{\mathbf{W}_{(0:z+1+l)}}(\mathbf{w}_{(0:z+1+l)}) \nonumber\\
&& \times   \mathbb{I}(y_{l;\mathbf{w}_{(0:z+1+l)}}=z+l) \bigg(c(l;\mathbf{w}_{(0:z+1+l)};0:z+l) \bigg)
\end{eqnarray}
\normalsize

Now, note that, the deployment upto the last relay does not depend on the shadowing 
in the links emanating from the locations $x > y_{l;\mathbf{w}_{(0:z+1+l)}}$. Hence, for any realization of the shadowing in all 
potential links, $c(l;\mathbf{w}_{(0:z+1+l)};0:k)=c(l;\mathbf{w}_{0:z+l};0:k)$ for all $k \leq z+l-1$. 
Note that $\mathbb{I}(\mathcal{E})$ 
is the indicator of the event $\mathcal{E}$; its vale is equal to $1$ if the event $\mathcal{E}$ occurs, or $0$ otherwise.

Hence, we obtain,

\footnotesize
\begin{eqnarray}
&& J_{\mu_{z+1}^{*}}(z+1;\mathbf{0})-J_{\mu_{z+1}^{*}}(z;\mathbf{0}) \nonumber\\
&=& \sum_{l=1}^{\infty}(1-\theta)^{l-1}\theta \sum_{\mathbf{w}_{(0:z+1+l)}} p_{\mathbf{W}_{(0:z+1+l)}}(\mathbf{w}_{(0:z+1+l)}) \nonumber\\
&& \times  \sum_{k=1}^{z+l-1} \mathbb{I}(y_{l;\mathbf{w}_{(0:z+1+l)}}=k) \bigg( c(l;\mathbf{w}_{(0:z+1+l)};k:z+1+l )  \nonumber\\
&&  -c(l;\mathbf{w}_{(0:z+1+l)};k:z+l ) \bigg) \nonumber\\
&&+ \sum_{l=1}^{\infty}(1-\theta)^{l-1}\theta \sum_{\mathbf{w}_{(0:z+1+l)}} p_{\mathbf{W}_{(0:z+1+l)}}(\mathbf{w}_{(0:z+1+l)}) \nonumber\\
&& \times   \mathbb{I}(y_{l;\mathbf{w}_{(0:z+1+l)}}=z+l)  c(l;\mathbf{w}_{(0:z+1+l)};z+l:z+1+l ) \nonumber\\
&& \label{eqn:term_2}
\end{eqnarray}
\normalsize

Now,

\footnotesize 
\begin{eqnarray*}
&&  \sum_{\mathbf{w}_{(0:z+1+l)}} p_{\mathbf{W}_{(0:z+1+l)}}(\mathbf{w}_{(0:z+1+l)}) \mathbb{I}(y_{l;\mathbf{w}_{(0:z+1+l)}}=k) \nonumber\\
&& \bigg( c(l;\mathbf{w}_{(0:z+1+l)};k:z+1+l ) -c(l;\mathbf{w}_{(0:z+1+l)};k:z+l ) \bigg)  \nonumber\\
\end{eqnarray*}
\begin{eqnarray}
&=& \sum_{\mathbf{w}_{(0:z+1+l)}} p_{\mathbf{W}_{(0:z+1+l)}}(\mathbf{w}_{(0:z+1+l)}) \mathbb{I}(y_{l;\mathbf{w}_{(0:z+1+l)}}=k) \nonumber\\
&& \bigg( \min_{\gamma \in \mathcal{S}} (\gamma+ \xi_o P_{out}(z+1+l-k,\gamma,w_{z+1+l,z+1+l-k}) \nonumber\\
&& - \min_{\gamma \in \mathcal{S}} (\gamma+ \xi_o P_{out}(z+l-k,\gamma,w_{z+l,z+l-k}))  \bigg) \label{eqn:term_1}
\end{eqnarray}
\normalsize

Since $P_{out}(r,\gamma,w)$ is increasing in $r$ for each $\gamma,w$, we must have the expression in 
(\ref{eqn:term_1}) greater than or equal to $0$. Also $c(l;\mathbf{w}_{(0:z+1+l)};z+l:z+1+l )$ in 
(\ref{eqn:term_2}) is always nonnegative. Hence, $J_{\mu_{z+1}^{*}}(z+1;\mathbf{0}) \geq J_{\mu_{z+1}^{*}}(z;\mathbf{0})$. Now, 
since $J_{\mu_{z+1}^{*}}(z;\mathbf{0}) \geq J_{\mu_{z}^{*}}(z;\mathbf{0})$, the result follows.

\section{Average Cost Per Step: With and Without Backtracking}\label{appendix:backtracking_average_cost}

\textbf{Proof of Theorem~\ref{thm:smdp-cost-vs-xi}}
Recall the definition of the functions $\mu^{(1)}$ and $\mu^{(2)}$. Now, 
$\frac{\xi_r+\sum_{\underline{w}} g(\underline{w}) \bigg(\mu^{(2)}(\underline{w})+ \xi_o P_{out}(\mu^{(1)}(\underline{w}),\mu^{(2)}(\underline{w}),w_{\mu^{(1)}(\underline{w})})\bigg)} 
{ \sum_{\underline{w}} g(\underline{w}) \mu^{(1)}(\underline{w}) ×}$ is the average cost of a specific  
stationary deterministic policy $\mu$ (by the Renewal Reward Theorem, 
since the placement process regenerates at each placement point). Hence, 

\footnotesize
 \begin{eqnarray*}
\lambda^{*} = \inf_{\mu}\frac{\xi_r+\sum_{\underline{w}} g(\underline{w}) \bigg(\mu^{(2)}(\underline{w})+ \xi_o P_{out}(\mu^{(1)}(\underline{w}),\mu^{(2)}(\underline{w}),w_{\mu^{(1)}(\underline{w})})\bigg)} 
{ \sum_{\underline{w}} g(\underline{w}) \mu^{(1)}(\underline{w}) ×} 
\end{eqnarray*}
\normalsize

For each policy $(\mu^{(1)},\mu^{(2)})$, the numerator is linear, increasing in $\xi_r$ and $\xi_o$ 
and the denominator is independent of $\xi_r$ and $\xi_o$. The proof follows immediately 
since the pointwise infimum of increasing, linear functions of $\xi_r$ and $\xi_o$ is increasing and 
jointly concave in $\xi_r$ and $\xi_o$.

\textbf{Proof of Theorem~\ref{theorem:comparison-backtracking-no-backtracking}}
Note that for the average cost problem with no backtracking, there exists an optimal threshold policy 
(similar to Theorem~\ref{theorem:policy_structure_sum_power_sum_outage_no_backtracking}), since the optimal policy for problem 
(\ref{eqn:sum_power_discounted_no_backtracking}) achieves $\lambda'$ average cost per step for $\theta$ sufficiently close to $0$. 
So, let one such optimal policy be given by the set of thresholds $\{c_{th}(r)\}_{A+1 \leq r \leq A+B-1}$. 

Now, let us consider the average cost minimization problem with backtracking. Consider the policy where we 
first measure $w_{A+1},w_{A+2},\cdots,w_{A+B}$ and decide 
to place a relay $u$ steps away from the previous relay (where $A+1 \leq u \leq A+B-1$) if 
$\min_{\gamma \in \mathcal{S}} (\gamma+\xi_o P_{out}(r,\gamma,w_r)) > c_{th}(r)$ for all $r \leq (u-1)$ and 
$\min_{\gamma \in \mathcal{S}} (\gamma+\xi_o P_{out}(u,\gamma,w_u)) \leq c_{th}(u)$. 
We must place if we reach at a distance $(A+B)$ from the previous relay. 
But this is a particular policy for the problem where we gather 
$w_{A+1},w_{A+2},\cdots,w_{A+B}$ and then decide where to place the relay, and clearly the average cost per step 
for this policy is $\lambda'$ which cannot be less than the optimal average cost $\lambda^*$.

\bibliographystyle{unsrt}
\bibliography{arpan-techreport}

\end{document}